\documentclass[reqno,12pt,letterpaper]{amsart}
\usepackage{amsmath,amssymb,amsthm,graphicx,mathrsfs,url,bbm,mathtools,cases,enumitem}
\usepackage[usenames,dvipsnames]{color}
\usepackage[colorlinks=true,linkcolor=Red,citecolor=Green]{hyperref}
\usepackage{amsxtra}
\usepackage[toc,page]{appendix}
\usepackage{wasysym} 
\usepackage{graphicx}
\usepackage{subcaption}
\usepackage{comment}
\def\arXiv#1{\href{http://arxiv.org/abs/#1}{arXiv:#1}}

\usepackage{soul}

\newcommand{\aw}[1]{\textcolor{blue}{AW: #1}}

\usepackage{mathtools}

\def\?[#1]{\textbf{[#1]}\marginpar{\Large{\textbf{??}}}}

\let\epsilon=\varepsilon 

\newcommand{\RR}{{\mathbb R}}
\newcommand{\NN}{{\mathbb N}}
\newcommand{\CC}{{\mathbb C}}

\newcommand{\ZZ}{{\mathbb Z}}

\newcommand{\tu}{\Tilde{u}}
\newcommand{\tv}{\Tilde{v}}

\newcommand{\td}{\widetilde}
\newcommand{\cA}{\mathcal{A}}

\newtheorem{theo}{Theorem}
\newtheorem{prop}{Proposition}[section]	
\newtheorem{defi}[prop]{Definition}
\newtheorem{ex}{Example}

\newtheorem{lemm}[prop]{Lemma}
\newtheorem{corr}[prop]{Corollary}
\newtheorem{rem}{Remark}

\numberwithin{equation}{section}

\newcommand{\Red}{\textcolor{red}}
\newcommand{\Blue}{\textcolor{blue}}

\DeclareMathOperator{\Spec}{Spec}

\let\Im=\Imag
\DeclareMathOperator{\loc}{loc}

\let\Re=\Real

\DeclareMathOperator{\Tr}{Tr}

\usepackage{scalerel}

\newcommand\reallywidehat[1]{\arraycolsep=0pt\relax%
\begin{array}{c}
\stretchto{
  \scaleto{
    \scalerel*[\widthof{\ensuremath{#1}}]{\kern-.5pt\bigwedge\kern-.5pt}
    {\rule[-\textheight/2]{1ex}{\textheight}} 
  }{\textheight} %
}{0.5ex}\\           
#1\\                 
\rule{-1ex}{0ex}
\end{array}
}
\author[S.\,Becker]{Simon Becker}
\address[Simon Becker]{ETH Zurich, 
Institute for Mathematical Research, 
R\"amistrasse 101, 8092 Zurich, 
Switzerland}
\email{simon.becker@math.ethz.ch}

\author[S.\,Quinn]{Solomon Quinn}
\address[Solomon Quinn]{School of Mathematics, University of Minesota, Minneapolis, MN, 55455, USA}
\email{srquinn@umn.edu}

\author[Z.\,Tao]{Zhongkai Tao}
\address[Zhongkai Tao]{Department of Mathematics, University of California, Berkeley, CA 94720, USA}
\email{ztao@math.berkeley.edu}

\author[A.\,Watson]{Alexander Watson}
\address[Alexander Watson]{School of Mathematics, University of Minesota, Minneapolis, MN, 55455, USA}
\email{abwatson@umn.edu}

\author[M.\,Yang]{Mengxuan Yang}
\address[Mengxuan Yang]{Department of Mathematics, University of California, Berkeley, CA 94720, USA}
\email{mxyang@math.berkeley.edu}

\title[Dirac cones and magic angles in TBG]{Dirac cones and magic angles in the Bistritzer--MacDonald TBG Hamiltonian}

\begin{document}
\begin{abstract}
We demonstrate the generic existence of Dirac cones in the full Bistritzer--MacDonald Hamiltonian for twisted bilayer graphene. Its complementary set, when Dirac cones are absent, is the set of \emph{magic angles}. We show the stability of magic angles obtained in the chiral limit by demonstrating that the perfectly flat bands transform into quadratic band crossings when perturbing away from the chiral limit. Moreover, using the invariance of Euler number, we show that at magic angles there are more band crossings beyond these quadratic band crossings. This is the first result showing the existence of magic angles for the full Bistritzer--MacDonald Hamiltonian and solves Open Problem No.\,2 proposed in the recent survey \cite{Z23}.
\end{abstract}
\maketitle 
\section{Introduction}
\label{sec:intro}
The discovery of superconductivity and strong electronic interactions in twisted bilayer graphene (TBG) has established TBG as an important platform for investigating strongly correlated physics within a system characterized by a topologically non-trivial band structure.

The Bistritzer-MacDonald (BM) Hamiltonian serves as the standard model for describing the effective one-particle band structure of TBG \cite{BM11}. Considerable efforts in both mathematical and physics literature have been devoted to analyzing specific properties of a notable limit of this model, known as the \emph{chiral limit} \cite{magic}.  Within this chiral limit, a discrete set of parameters has been identified \cite{magic,beta}, at which the bands closest to zero energy become entirely flat. This phenomenon is often interpreted as a key factor contributing to the system's strongly correlated electronic properties.

It has been shown in \cite{magic,beta,Z23} that in the chiral limit, the band structure exhibits either a Dirac cone (i.e., a conic singularity of the band structure) at zero energy or a flat band. In short, the model exhibits Dirac cones if and only if the angle is not \emph{magic}. 

Since completely flat bands at zero energy are not observed away from the chiral limit, the absence of Dirac cones is studied to identify magic angles in the band structure of the self-adjoint BM Hamiltonian \cite{BM11}. The BM Hamiltonian is defined as $H (\alpha,\lambda): H^1 ( \mathbb C ; \mathbb C^4 ) \subset L^2 ( \mathbb C ; \mathbb C^4 ) \to 
L^2 ( \mathbb C ; \mathbb C^4 )$ given by
\begin{equation}
\label{eq:defBM} H (\alpha,\lambda) := \begin{pmatrix} \lambda C & D (\alpha)^* \\
D (\alpha) &  \lambda C \end{pmatrix}
\end{equation}
with
\[ D(\alpha):=\begin{pmatrix} 2D_{\bar z} & \alpha U(z) \\ \alpha U(-z) & 2D_{\bar z}. \end{pmatrix} \text{ and }C := \begin{pmatrix} 0 & V(z) \\ \overline{V(z)} & 0 \end{pmatrix}\]
and tunnelling potentials $U,V$ satisfying symmetries for $\gamma \in \Lambda = \ZZ + \omega \ZZ$
\[\begin{split}
U(z+\gamma) = e^{i \langle \gamma,K \rangle} U(z), \quad U(\omega z) = \omega U(z), \text{ and } \overline{U(\overline{z})} = -U(-z) \\
V(z)=V(\overline{z}) = \overline{V(-z)}, \quad  V(\omega z) = V(z), \text{ and }V(z+\gamma)=e^{i \langle \gamma,K \rangle} V(z),
\end{split}\]
where $\omega = e^{2\pi i/3}$ and $K= \frac{4}{3}\pi$ with inner-product $\langle a,b\rangle:=\Re(a \bar b)$ for $a,b \in \CC.$
Honeycomb lattices consist of two non-equivalent vertices per fundamental domain that we denote by $A$ and $B$. Then parameters $\alpha$ and $\lambda$ tune the strengths of the tunnelling potentials of the $A/B$ and $A/A,B/B$ regions, respectively and are inversely proportional to the twisting angle $\theta$ for $\theta \ll 1.$ The chiral limit is obtained by setting $\lambda = 0$ in the Hamiltonian \eqref{eq:defBM}, so $\lambda$ also represents how far the model is from this limit. We refer the reader to Section \ref{sec:bm} for a detailed discussion of the special form of this Hamiltonian and additional definitions on the tunnelling potentials.
The Hamiltonian \eqref{eq:defBM} commutes with the translation operator
\begin{equation}
\label{eq:FL_ev}   \mathscr L_{\gamma } \mathbf w ( z ) := 
\operatorname{diag}(  \omega^{\gamma_1 + \gamma_2}  , 1 , \omega^{\gamma_1 + \gamma_2}  , 1) 
\mathbf w( z + \gamma ),  \ \ \ \gamma \in \Lambda, \quad \gamma=\gamma_1 + \omega  \gamma_2, \ (\gamma_1,\gamma_2) \in \ZZ^2.  
\end{equation}
Thus, by Bloch--Floquet theory, the Hamiltonian is equivalent to a family of Bloch transformed operators
\begin{gather}
    \label{eq:floH}
    H_k(\alpha,\lambda) = \begin{pmatrix} \lambda C & D(\alpha)^* + \bar k \\ D(\alpha)+k & \lambda C \end{pmatrix}: L^2_{0}(\CC/\Lambda;\CC^4) \to L^2_{0}(\CC/\Lambda;\CC^4),\\ 
    \label{eq:defscH} 
    L^2_{k}(\CC/\Lambda;\CC^4):= \{ \mathbf v \in L^2_{\loc} ( \mathbb C, \CC^4 ) :
\mathscr L_{\gamma } \mathbf v = e^{ i \langle k , \gamma \rangle } \mathbf v , \ \ \gamma \in \Lambda \} , \quad k \in \CC/\Lambda^*
\end{gather}
such that $$\Spec_{L^2 ( \mathbb C ; \mathbb C^4 ) }(H(\alpha,\lambda))= \bigcup_{k \in \CC/\Lambda^*} \Spec_{L^2_{0}(\CC/\Lambda;\CC^4)} (H_{k}(\alpha,\lambda)),$$
where $\Lambda^*$ is the dual lattice associated with $\Lambda.$
The spectrum of $H_{k}(\alpha,\lambda)$ on $L^2_{0}(\CC/\Lambda;\CC^4)$ is discrete, and we label it as follows
\begin{equation}
\label{eq:eigs} 
\begin{gathered} \{ E_{ \pm j } (\alpha,\lambda; k ) \}_{ j \in  \mathbb N_+ },\ \  E_{\pm 1} ( \alpha,\lambda; -K) =  E_{\pm 1} ( \alpha,\lambda;  0 ) = 0, \\
 \ \cdots \leq E_{-2} ( \alpha,\lambda; k ) \leq E_{-1} ( \alpha,\lambda; k ) \leq 0\leq E_1 ( \alpha,\lambda; k ) \leq E_2 ( \alpha,\lambda; k ) \leq \cdots,
\end{gathered} \end{equation}
forming the \emph{Bloch bands}. The points $k= -K, 0 $
are high-symmetry points and are typically denoted by $ K $ and $ K' $ in the physics literature\footnote{We shift them by $K$ for notational convenience in the computation.}. At points $k=0,-K$, there exist two protected states (cf.\,\cite[Proposition 2]{bz23})
\begin{equation}
\label{eq:protect}
    \varphi_{k}(\alpha,\lambda) = (u,v)^T,\ \psi_{k} (\alpha,\lambda) = (\tu, \tv)^T\in \ker_{L^2_0(\CC/\Lambda;\CC^4)} (H_k(\alpha,\lambda)),\ k=0,-K
\end{equation}
for all $\lambda,\alpha \in \CC$ explaining $E_{\pm 1} ( \alpha,\lambda; -K) =  E_{\pm 1} ( \alpha,\lambda;  0 ) = 0$ in \eqref{eq:eigs}. As we shall argue, the protected states play an essential role in analyzing the presence of Dirac cones in this model.

As mentioned above, we shall study the band structure of the BM Hamiltonian \eqref{eq:defBM} near the points $0, -K$ in the Brillouin zone $\CC/\Lambda^*$. In particular, we investigate the existence of conic singularities (see Figure \ref{fig:band}) in $E_{\pm 1}(\alpha, \lambda; k)$ for $(\alpha,\lambda)\in \RR^2$. First, we introduce the following 
\begin{defi}[Simple Dirac cone]
\label{def:cone}
    Assume $\dim \ker_{L^2_{0}(\CC/\Lambda;\CC^4)}(H_{0}(\alpha,\lambda)) =2$.
    We say that the BM Hamiltonian \eqref{eq:defBM} exhibits a simple Dirac cone at $k=0$ at $(\alpha,\lambda)$ if and only if 
    \[E_{\pm 1}(\alpha,\lambda;k)=\pm |v_F(\alpha,\lambda)| \vert k\vert+ \mathcal O(\vert k \vert^2) \text{ with } v_F(\alpha,\lambda) \neq 0.\]
    We call such parameters $(\alpha,\lambda)$ \emph{non-magic}. When $v_F(\alpha,\lambda)=0$, i.e. the dispersion relation is not linear, we call the parameters $(\alpha,\lambda)$ \emph{magic}.
\end{defi}
\begin{rem}
    The Dirac cone at $k=-K$ is defined analogously. To simplify the presentation, we prove all results near $k=0$. In fact, the two points are connected by the symmetries $\mathscr S$ and $\mathscr M$ that will be defined in \eqref{eq:mapp}.
\end{rem}

\begin{figure}[ht]
    \includegraphics[width=10cm]{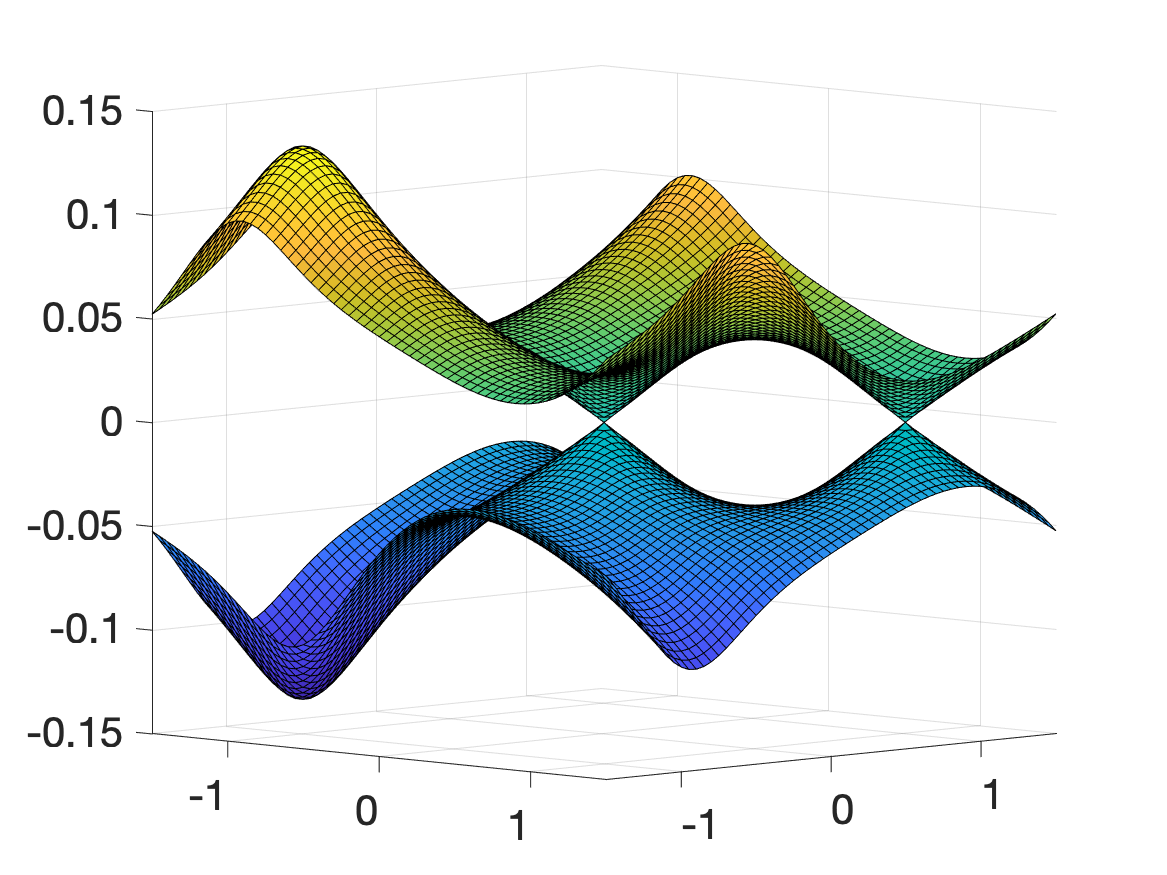}
    \caption{The two bands $E_{\pm 1}(\alpha,\lambda;k)$ closest to zero energy of BM Hamiltonian exhibiting two Dirac points for parameters $\alpha=0.7$ and $\lambda=0.3.$}
    \label{fig:band}
\end{figure}

We start by proving a theorem on the generic existence (cf.~Figure \ref{fig:fermi}) of Dirac cones for the bands $E_{\pm 1}(\alpha,\lambda,k)$ near $k=0$.
\begin{theo}[Generic existence of Dirac cones]
\label{thm:generic}
    There is a locally finite family of points and analytic curves $\{S_i\}$ such that the BM Hamiltonian \eqref{eq:defBM} exhibits a simple Dirac cone at $k=0$ for all $(\alpha,\lambda)\in \RR^2\setminus \bigcup_i S_i$. 
\end{theo}

\begin{figure}[ht]
    \includegraphics[width=8cm,height=7cm]{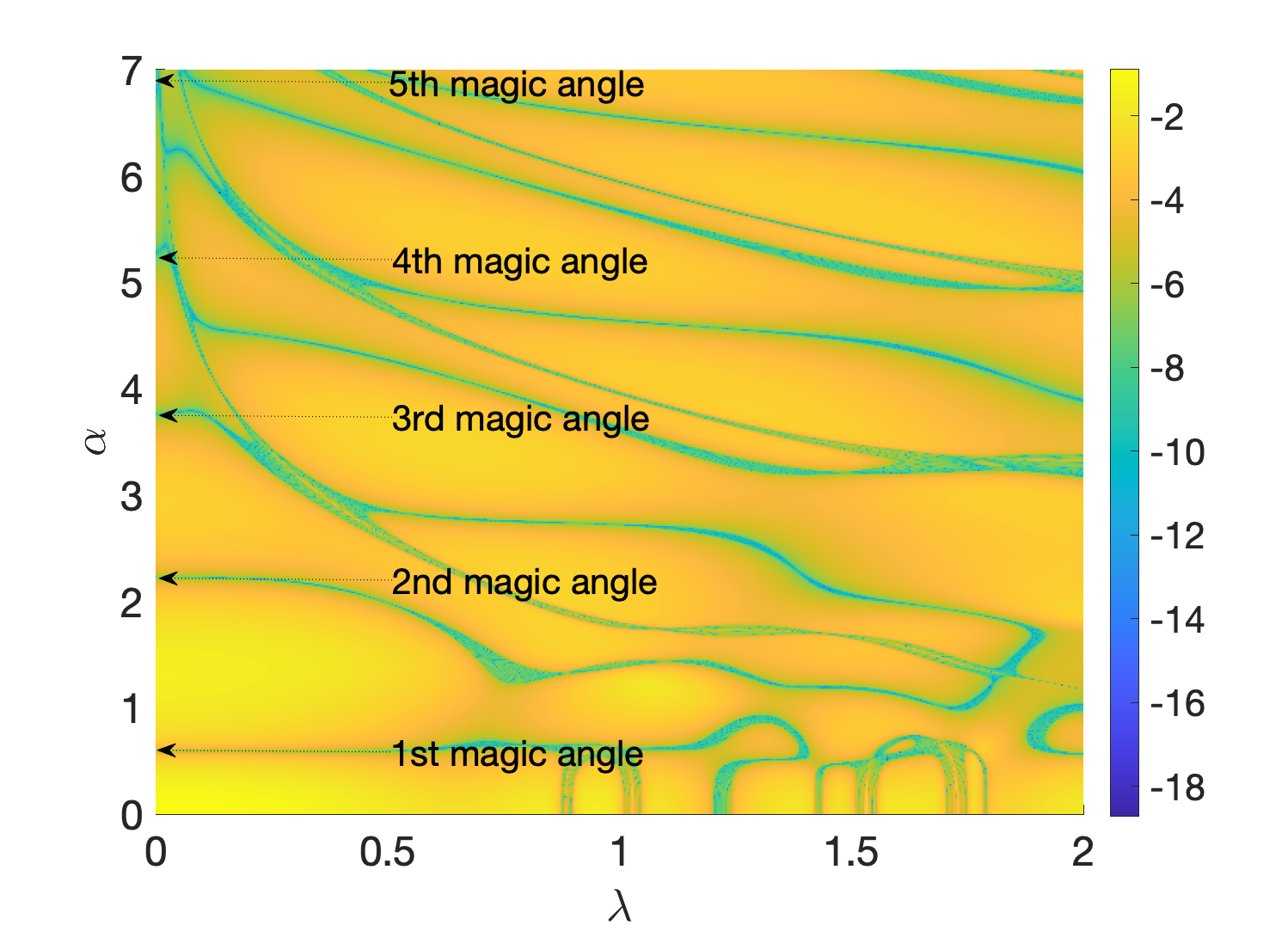}
    \includegraphics[width=8cm,height=7cm]{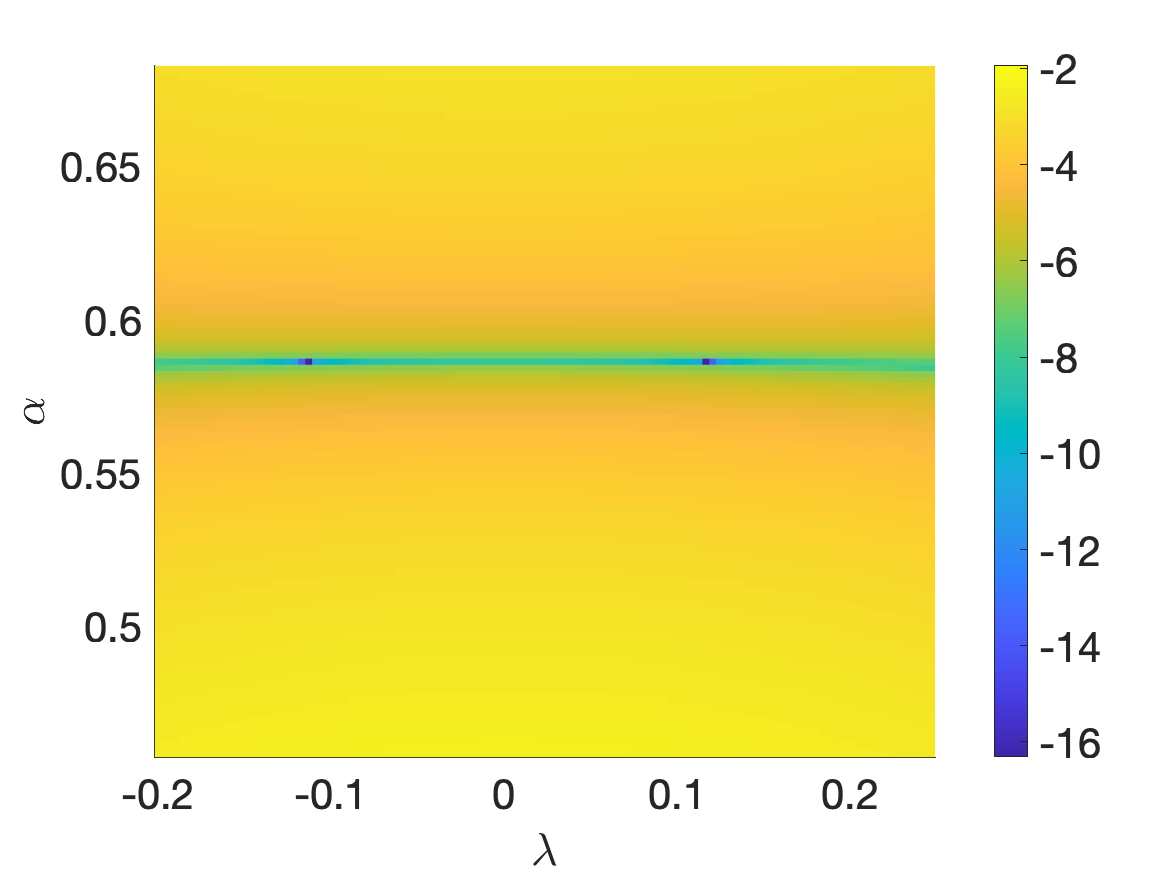}
\caption{ Logarithm of Fermi velocity  $(\alpha, \lambda)\mapsto \log |v_F(\alpha,\lambda)|$. Yellow colors indicate the existence of Dirac cones for the correspoding parameter $(\alpha,\lambda)\in\RR^2$, whereas blue colors indicate the absence of Dirac cones or extremely small Fermi velocities $v_F(\alpha,\lambda)$. The second picture zooms in near the first magic angle $(\alpha,\lambda)\sim (0.588,0)$ of the chiral limit illustrating Theorem \ref{thm:generic}.}
\label{fig:fermi}
\end{figure}

Next we show that there exist real parameters $(\alpha,\lambda)\in\RR^2$ at which the Dirac cones disappear, as suggested by Figure \ref{fig:fermi} numerically. 

Recall that in the chiral limit of the BM Hamiltonian $H(\alpha,0)$ where $\lambda = 0$ in \eqref{eq:defBM}, it has been proven (cf.~\cite{magic,beta,bhz2}) that there exists a discrete set $\mathcal{A}\subset \CC$ such that 
\[\alpha\in \mathcal{A} \iff E_{\pm 1} ( \alpha, 0; k ) = 0 \ \  \forall k\in \CC, \]
which means that there is a flat band at zero energy for $\alpha\in \mathcal{A}$. If $E_2(\alpha,0,k)$ does not vanish for all $k$ as well, then $\alpha$ is called \emph{simple}. The existence of the first real magic angle $\alpha\in \mathcal{A}\cap \RR$ is shown in \cite{walu,bhz1} for the BM potential \eqref{eq:defUV} used in the physics literature \cite{BM11}. It is shown in \cite[Appendix]{Z23} and \cite{magic} that $(\alpha,0)$ is non-magic in the sense of Definition \ref{def:cone} if and only if $\alpha\notin \mathcal{A}$. This means that for the chiral Hamiltonian $H(\alpha,0)$, the Dirac cone at $k=0$ disappears precisely for magic $\alpha \in \mathcal{A}$ at which the Hamiltonian exhibits a flat band. The next theorem generalizes this result to the BM Hamiltonian \eqref{eq:defBM}. The theorem shows that magic angles of the chiral model, persist in the full BM Hamiltonian. This is illustrated in Figure \ref{fig:Magic_angles}.
\begin{theo}[Persistence of chiral magic angles]
\label{thm:line}
  Assume $\partial_{\alpha} v_F(\alpha_0,0) \neq 0$ (cf.\,Proposition \ref{prop:real}) for some $\alpha_0\in\mathcal{A} \cap \mathbb R$ simple. There is $\epsilon>0$ and a real valued real analytic function $\lambda \mapsto \alpha(\lambda)$ for $\vert \lambda\vert < \epsilon$ with $\alpha(0)=\alpha_0$ and $\alpha'(0)=0$ such that $(\alpha,\lambda)$ is magic along the curve $ (\alpha(\lambda),\lambda)$  for $\lambda\in (-\epsilon,\epsilon)$.
  
 On the other hand, if $\alpha_0\notin\mathcal{A}$, then there exists a neighborhood $U_{\alpha_0}\subset \RR^2 $ of $(\alpha_0,0)$ such that $(\alpha,\lambda)$ is non-magic on $U_{\alpha_0}$.
\end{theo}
\begin{rem}
    For the standard BM potential \eqref{eq:defUV}, the condition $\partial_{\alpha} v_F(\alpha_0,0) \neq 0$ in Theorem \ref{thm:line} is numerically checked for the first six magic angles (see Table \ref{tab:c10}). 
\end{rem}

\begin{figure}[ht]
    \centering
    \includegraphics[width=8cm]{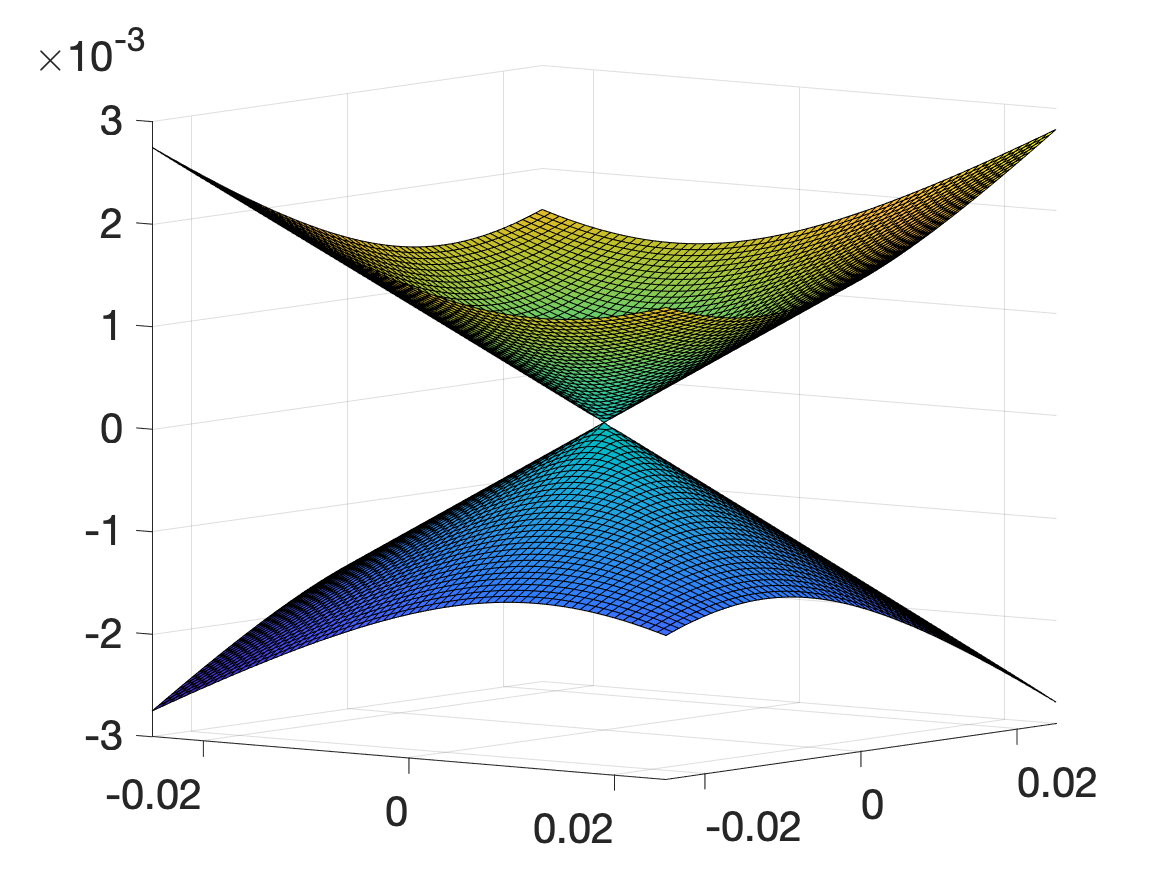}
    \includegraphics[width=8cm]{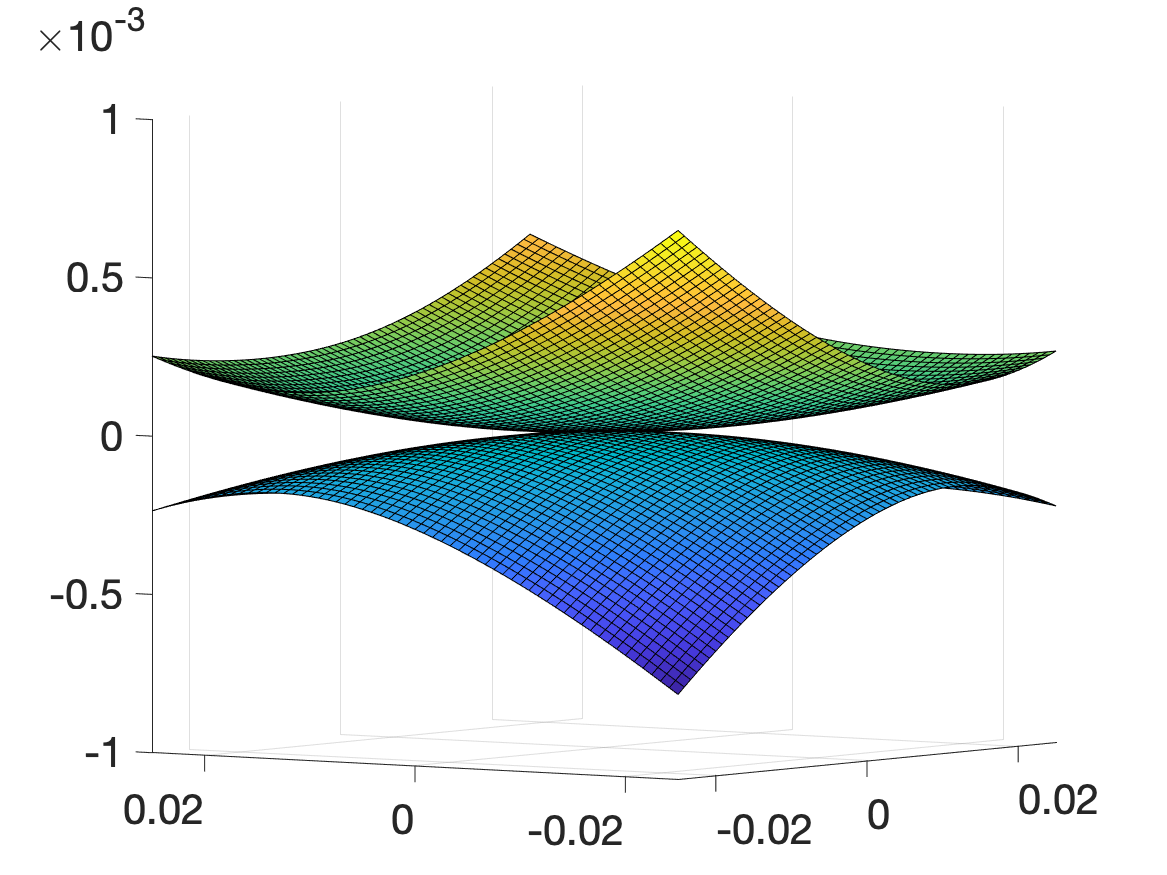}
    \caption{Non-magic band crossing for $\alpha=0.7$ and $\lambda=0.3$ with Dirac cone (left). Magic band crossing for $(\alpha,\lambda) = (0.62, 0.95)$, which is located on the blue curve emerging from the first magic angle in Figure \ref{fig:fermi} (right).}
    \label{fig:Magic_angles}
\end{figure}
In Section \ref{sec:topo}, we study band touching near magic parameters. We show that when transitioning from conic singularities to quadratic band touching at Dirac points $-K,0$, the change of winding number implies the existence of other band crossings at points other than $-K,0$ points as shown in Figure \ref{f:touching}. In fact, we prove the following
\begin{theo}
    For magic parameters $(\alpha,\lambda)$ near simple chiral magic angles $\alpha_0\in \mathcal{A}$ and $\lambda =0$, suppose the first two bands $E_{\pm 1}(\alpha,\lambda,k)$ are isolated from other bands and exhibit quadratic band crossing in the sense that near $k=0$
    \begin{equation*}
    E_{\pm 1}(\alpha,\lambda,k) = c_{\pm}(\alpha,\lambda)|k|^2 + \mathcal{O}(|k|^3), \quad  c_{\pm}(\alpha,\lambda)\neq 0. 
    \end{equation*}
    Then the two bands touch at some additional points other than $k = 0,-K$.
\end{theo}
\begin{rem}
    The assumption of the theorem can be weaken: it holds as long as there exists a path $\gamma$ in $\RR^2$ connecting $(\alpha_0,0)$ and $(\alpha,\lambda)$ such that the first two bands $E_{\pm 1}(\alpha,\lambda,k)$ are gapped from the rest of bands along $\gamma$. See Section \ref{sec:topo} for more details. 
\end{rem}
To understand the transition from conical intersection in the band structure to quadratic ones, we need to develop the study of topological phases for our model. The mechanism behind this involves a winding number known as the Euler number which classifies the triviality of real vector bundles of rank two. The reality condition appears in this model due to the $C_{2z}T$ symmetry. 

Here we only discuss the BM model \eqref{eq:defBM} but our argument can be generalized to studying band touchings of two-band systems gapped from the rest of the bands. 

See also \cite{cw24} for discussions of the splitting of quadratic band touchings into Dirac cones for Schr\"odinger operators and \cite{dynip} for the formation and splitting of quadratic band touchings by in-plane fields in twisted bilayer graphene.

The analysis of conical intersections in band structures has received a considerable attention and has been studied for honeycomb lattice structures \cite{fefferman2012honeycomb}, describing materials such as graphene, as well as optical and acoustic analogues \cite{am20}. The first results on the dispersion relation of tight-binding models for honeycomb structures have been obtained more than half a century ago (see~\cite{Wal_pr47,SloWei_pr58}). A mathematical model of honeycomb quantum graphs with even potential on edges for graphene is studied by Kuchment--Post \cite{KucPos_cmp07}. The Schr\"odinger operator $-\Delta + V$ with smooth real valued potentials $V$ exhibiting honeycomb symmetry has been studied by Grushin \cite{Gru_mn09} for small potentials using perturbation theory method. For a generic set of smooth potentials, the existence of Dirac cones has been established by Fefferman--Weinstein \cite{fefferman2012honeycomb, FefLeeWei_prep16} and later by Berkolaiko--Comech \cite{bc18}. This is generalized to potentials with singularities at honeycomb lattice points in the work of Lee \cite{lee}. For generalizations to different class of elliptic operators defined on the honeycomb lattice, see \cite{cw21,lwz,llz,ca}. See also the survey by Kuchment \cite{ku}. The omnipresence of conical singularities in the context of topological insulators has been analyzed by Drouot \cite{top}. For twisted bilayer graphene, in the chiral limit of the BM Hamiltonian, the existence of Dirac cone has been proved by some of the authors of this paper in \cite[Appendix]{Z23} and in \cite{magic}.

\subsection*{Structure of the paper}
\begin{itemize}
    \item In Section \ref{sec:sym}, we review properties of the BM Hamiltonian which includes a basic derivation and symmetries of the BM Hamiltonian and the existence of protected states. 
    \item In Section \ref{sec:cone}, we prove the generic existence of Dirac cones. 
    \item In Section \ref{sec:magic}, we extend the notion of magic angles, show the stability of magic angles obtained from the chiral limit to the general BM Hamiltonian, and prove the existence of magic angles. 
    \item In Section \ref{sec:topo}, we discuss Euler numbers and winding numbers of bands and use them to study the behavior and crossing of bands near magic parameters.
\end{itemize}
\subsection*{Acknowledgements}
The authors are very grateful to Maciej Zworski for many helpful discussions and proposing this project. The authors would also like to thank Gregory Berkolaiko, Patrick Ledwith, Qiuyu Ren and Oskar Vafek for helpful discussions. Their insights were key to the development of the paper. We would like to thank Jens Wittsten for allowing us to use the left figure in Figure \ref{fig:tunneling} that he created. ZT and MY were partially supported by the NSF grant DMS-1952939 and by the Simons Targeted Grant Award No. 896630. AW's research was supported by the NSF grant DMS-2406981. SB acknowledges support by the SNF Grant PZ00P2 216019.

\section{Symmetries of the Hamiltonian and protected states}
\label{sec:sym}
In this section we briefly review the physical origin, symmetries, and Bloch-Floquet theory of the Hamiltonian \eqref{eq:defBM}.
\subsection{BM Hamiltonian}
\label{sec:bm}

\begin{figure}[ht]
    \centering
    \includegraphics[width=6.5cm]{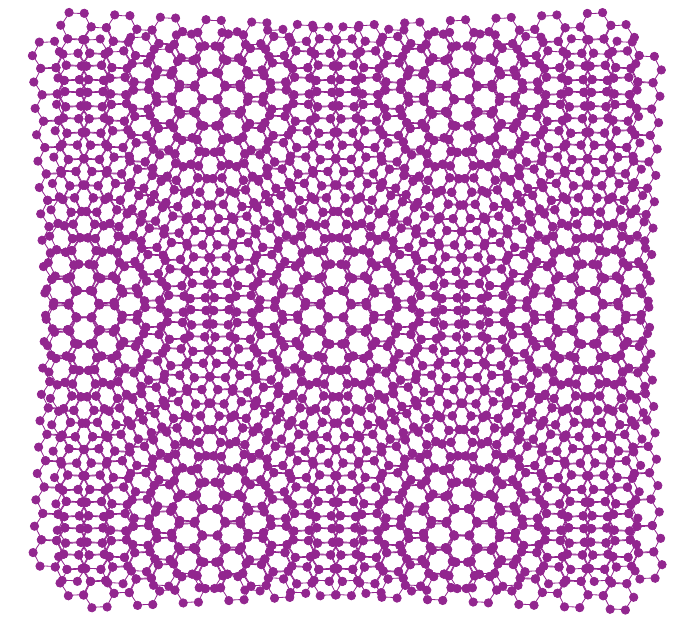}
    \includegraphics[width=7cm]{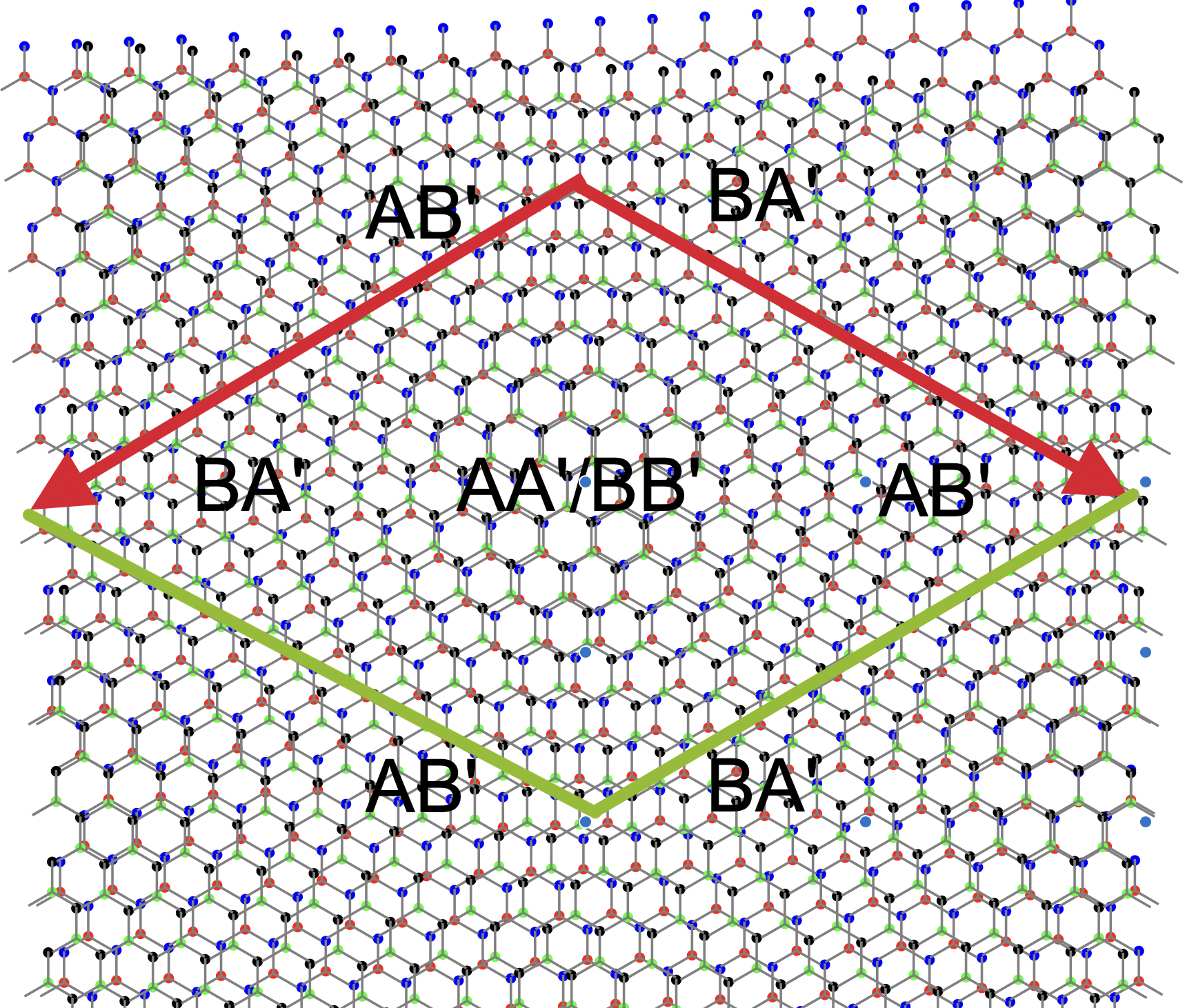}
    \caption{Left: Moir\'e patterns in TBG (courtesy of Jens Wittsten). Right: A moir\'e fundamental cell with regions of different ($AA', BB', AB'$... ) particle-type overlaps.}
    \label{fig:tunneling}
\end{figure}

Twisted bilayer graphene consists of two stacked graphene layers with a relative interlayer twist. The atomic structure of each graphene layer is the union of two offset Bravais lattices of carbon atoms, referred to as the $A$ and $B$ sublattices. The single-particle electronic properties of twisted bilayer graphene can be modeled by the following Bistritzer-MacDonald Hamiltonian $\mathcal{H} : L^2(\mathbb{R}^2;\mathbb{C}^4) \rightarrow L^2(\mathbb{R}^2;\mathbb{C}^4)$,
 \begin{equation}
 \label{eq:mathscrH}
 \mathcal H(\beta,\zeta)  = \begin{pmatrix} -i\boldsymbol{\sigma}\cdot\nabla & T(\beta,\zeta \mathbf r)  \\  
T(\beta,\zeta \mathbf r ) ^\dag & -i\boldsymbol{\sigma}\cdot \nabla \end{pmatrix},
\end{equation}
acting on functions $\psi(\boldsymbol{r}) = (\psi_1^A,\psi_1^B,\psi_2^A,\psi_2^B)^\top$, $\boldsymbol{r} = (x,y)$, where $|\psi_i^\sigma|^2$ represents electronic probability densities on layer $i$ and sublattice $\sigma$. 

In the diagonal (intralayer) terms, we have $\boldsymbol{\sigma} = (\sigma_x,\sigma_y)$, where $\sigma_\bullet $ are the Pauli matrices, and $\boldsymbol{\sigma}\cdot\nabla = \sigma_x \partial_x + \sigma_y \partial_y$. These terms capture the effective Dirac dispersion for electrons in monolayer graphene near to the Fermi level. The off-diagonal terms describe interlayer tunnelling, which is approximated by the interlayer tunnelling matrix potential
\[
    T(\beta,\zeta \mathbf r) = \begin{pmatrix}\beta_{AA} V(\zeta \mathbf r) &\beta_{AB} \overline{U(-\zeta \mathbf r)} \\  \beta_{BA} U(\zeta \mathbf r) & \beta_{AA} V(\zeta \mathbf r) \end{pmatrix}.
\]
Here, $\beta = (\beta_{AA},\beta_{AB})$ are the tunnelling potential amplitudes between like and unlike sublattices, $\zeta$ is the relative twist angle (here we assume $\zeta$ is small so that we can use the approximation $\sin(\zeta) \approx \zeta$), and $U$ and $V$ are given by

\[ 
\begin{gathered} 
    U(\mathbf r) = \sum_{i=0}^2 \omega^i e^{-i q_i\cdot \mathbf r}, \ \ 
    V(\mathbf r) = \sum_{i=0}^2  e^{-i q_i\cdot \mathbf r}, \ \ 
    q_i : = R^i  \begin{pmatrix} \ \ 0 \\ -1\end{pmatrix} , \ \  
    R := \tfrac{1}{2}\begin{pmatrix} -1 & -\sqrt{3} \\ \sqrt{3} & -1 \end{pmatrix},
\end{gathered} 
\]
where $\omega = e^{2 \pi i/3}$. The model \eqref{eq:mathscrH} was introduced in \cite{BM11}; for its mathematical justification and detailed discussion of the various approximations involved in deriving \eqref{eq:mathscrH}, see \cite{Wa22,Ko24,Q24,Ca23}. If we denote the lattice of periodicity of the Hamiltonian \eqref{eq:mathscrH} by $\Gamma$, the Hamiltonian \eqref{eq:mathscrH} commutes with translations in the moir\'e lattice $\frac{1}{3}\Gamma$ up to cubic root of unity phases; see e.g. \cite{Wa22} for details.

    The Hamiltonian \eqref{eq:defBM} is obtained from \eqref{eq:mathscrH} as follows. Let $K = \operatorname{diag}(1,\sigma_x,1)$. Then, upon making the change of variables $\zeta \mathbf r \mapsto \mathbf r $, $\lambda = \beta_{AA} /\zeta$, and $\alpha = \beta_{AB} /\zeta$, we find
\[ \zeta^{-1} K  \mathcal H(\beta,\zeta) K = \begin{pmatrix} \lambda C & D(\alpha)^{*} \\  D(\alpha) & \lambda C \end{pmatrix},\]
where 
\[\begin{split} 
    D(\alpha) &= \begin{pmatrix}D_{x} + i D_{y} & \alpha U(\mathbf r) \\ \alpha U(-\mathbf r) & D_{x}+i D_{y}\end{pmatrix} \text{ and }
C =\begin{pmatrix}0 & V(\mathbf r) \\ V(-\mathbf r) & 0 \end{pmatrix}, 
\end{split}\] 
which now acts on functions $\psi(\boldsymbol{r}) = (\psi_1^A,\psi_2^A,\psi_1^B,\psi_2^B)^\top$. Changing coordinates from $\boldsymbol{r} = (x,y)$ to $z$ so that $x + i y = \frac{4}{3} \pi i z$ we obtain \eqref{eq:defBM}, where $D(\alpha)$ and $C$ are given by (we abuse notation to write functions of $\boldsymbol{r}$ and $z$ by the same letters)
\begin{equation}
    \label{eq:defCC}
    D(\alpha)=\begin{pmatrix} 2D_{\bar z} & \alpha U(z)\\ \alpha U(-z) & 2D_{\bar z} \end{pmatrix} \text{ and }
    C = \begin{pmatrix} 0 & V ( z )  \\
V  ( -z ) & 0 \end{pmatrix},
\end{equation}
where
\begin{equation}
\label{eq:defUV} 
U(z) = \sum_{k=0}^2 \omega^k e^{\frac{1}{2}(z \bar \omega^k - \bar z \omega^k)} \text{ and }
V(z) = \sum_{k=0}^2 e^{\frac{1}{2}(z \bar \omega^k - \bar z \omega^k)}. 
\end{equation}
The moir\'e lattice in these coordinates is $\Lambda = a_1 \ZZ +a_2 \ZZ ,$ where $a_{\ell} = \omega^{\ell-1}$. This lattice and its associated reciprocal lattice $\Lambda^*$ are related to those of the physical space moir\'e by 
\[ \tfrac{1}{3}\Gamma=\tfrac{4}{3}\pi i \Lambda, \quad 3\Gamma^*=\frac{3}{4\pi i }\Lambda^*. \]

As discussed above, throughout this work we make the more general assumption that $U$ and $V$ are arbitrary smooth functions satisfying the symmetries 
\begin{equation}
\begin{gathered}
 \label{eq:UV}  
 \overline{U(z)}=U(\bar z), \ \ U(\omega z)=\omega U(z), \ \ 
 U ( z + \gamma ) = \bar{ \omega}^{\gamma_1+\gamma_2} U ( z) \\
 V (  z ) = V ( \bar z ) = \overline{ V ( - z ) } ,  \ \ V ( \omega z ) = V ( z ) , 
 \ \ V ( z + \gamma ) = \bar{ \omega}^{\gamma_1+\gamma_2} V ( z), 
 \end{gathered}
\end{equation}
with $\gamma= \gamma_1 a_1+ \gamma_2 a_2 \in \Lambda.$ Since our proofs rely only on symmetries of the model \eqref{eq:defBM}, this generalization does not complicate our proofs, but we use \eqref{eq:defUV} for our numerical computations for simplicity. 

    \begin{rem}
The assumption \eqref{eq:defUV} on the potentials $U,V$ amounts to an approximation where momentum-space hopping is truncated to nearest-neighbors in the moir\'e reciprocal lattice. This is generally an excellent approximation since the non-zero spacing between the graphene sheets causes the magnitude of momentum-space hops to decay exponentially with distance \cite{Wa22,BM11}. Note that mechanical relaxation \cite{Ca18,Ca20} complicates this picture by enhancing the strength of longer-range momentum hops; see e.g. \cite{Ma23}.
    \end{rem}

    \begin{rem}
    It is natural to ask whether our results could be generalized beyond the other simplifying approximations implicit in the model \eqref{eq:mathscrH}. One of these approximations is neglecting the rotation of the monolayer Dirac cones, see, e.g.~\cite[equation (8)]{BM11}. Another is the neglect of derivative terms in the interlayer tunneling, see e.g. \cite{Ca19}. Since these terms preserve the translation and rotation symmetries of the model, and the $\mathcal{P} \mathcal{T}$ symmetry defined below, all of our results should apply to models including these terms. However, we do not consider these generalizations in the present work because the statement of our results would become much more complicated. This is because these terms break the particle-hole symmetry $\mathscr{S}$ defined below, so that there is no need for Dirac cones to occur at energy $0$. Our results would then have to allow for Dirac cones to occur in a suitable range of energies.
    \end{rem}

\subsection{Symmetries revisited}
\label{sec:symmetry}
We start by recalling the basic translational and rotational symmetries of \eqref{eq:defBM}. The modified translation operator $ \mathscr L_\gamma,$ defined in \eqref{eq:FL_ev}, 
commutes with the Hamiltonian $ H( \alpha,\lambda ) $. 
We also define the rotation operator
\[   \Omega u ( z ) := u ( \omega  z) , \ \  u \in \mathscr S' ( \mathbb C; \mathbb C^2 ) ,\ \ \Omega D ( \alpha ) = \omega D ( \alpha ) \Omega.  \]
We then extend it to a commuting action with $ H ( \alpha,\lambda ) $ by introducing
\begin{equation}
\label{eq:defC} 
\mathscr C := \begin{pmatrix}
\Omega &  \ 0 \\
 0 & \bar \omega \Omega \end{pmatrix} : L_{\rm{loc}} ^2 ( \mathbb C ; \mathbb C^4 ) \to L^2_{\rm{loc}} ( \mathbb C ; \mathbb C^4 )  \end{equation}
such that
\[  
 \mathscr C H ( \alpha,\lambda ) = H( \alpha,\lambda ) \mathscr C , \ \ 
\mathscr L_\gamma \mathscr C = \mathscr C \mathscr L_{\omega \gamma}, 
\ \  \mathscr C \mathscr L_\gamma  =  \mathscr L_{\bar \omega \gamma} \mathscr C.
\]
From these two symmetries, it is possible \cite[(3.10),(3.15)]{dynip} to obtain the following orthogonal decomposition 
\begin{equation}
\label{eq:ortho} 
\begin{gathered}
 L^2_k ( \mathbb C/\Lambda ) = \bigoplus_{ p \in \mathbb Z_3 } 
 L^2_{k,p}  ( \mathbb C/\Lambda ), \ \ \ 
 k \in \mathcal K/\Lambda^* , 
 \end{gathered}
 \end{equation}
 where
 \begin{equation}
 \label{eq:defK} 
 \  \mathcal K := \{ k \in \mathbb C: \omega k
 \equiv k \!\! \mod \Lambda^* \} =  \{ K, -K, 0 \} + \Lambda^*, \quad  K = \frac 4 3 \pi
 \end{equation}
 and 
 \begin{equation}
\label{eq:defLpk} 
\begin{gathered} 
L^2_{ k,p} ( \mathbb C/\Lambda ; \mathbb C^4 ) := 
\{ u \in L^2_{\rm{loc} } ( \mathbb C ; \mathbb C^4 ) : 
\mathscr L_\gamma \mathscr C^\ell  u = e^{ i \langle k , \gamma \rangle } 
\bar  \omega^{\ell p} u \}, \\  
k \in (\tfrac13 \Lambda^*)/\Lambda^* \simeq \mathbb Z^2_3, \ \  p \in \mathbb  Z_3 . 
 \end{gathered}
 \end{equation}

\subsection{Additional symmetries}
We recall additional symmetries of the BM Hamiltonian using the notation of \cite{bz23} that are central in proving the existence of Dirac cones. We start with the symmetries of $ D(\alpha) $, the non-self-adjoint building block of the BM Hamiltonian \eqref{eq:defBM} for $ U $ satisfying the conditions in \eqref{eq:UV}. 

We recall an anti-linear symmetry,
\begin{equation}
 \label{eq:defQ} 
 {\mathscr Q}v(z) := \overline{ v(-z)}
 , \ \ \ \ {\mathscr Q}(D(\alpha) + k){\mathscr Q} = (D(\alpha) + k)^*, 
 \end{equation}
and two linear symmetries,
\begin{equation}
\label{eq:defH}  \mathscr H \begin{pmatrix} u_1 ( z )\\u_2 ( z ) \end{pmatrix}  :=  \begin{pmatrix} 
- i u_2 ( - z ) \\ i u_1 ( - z) \end{pmatrix} , \ \ \ \   \mathscr H ( D ( \alpha ) + k )  \mathscr H = - ( D( \alpha ) 
- k ) , \end{equation}
\begin{equation}
\label{eq:defN}
\mathscr N \begin{pmatrix}
    u_1(z)\\
    u_2(z)
\end{pmatrix} := \begin{pmatrix} u_2 ( - \bar z ) \\ u_1 ( - \bar z ) \end{pmatrix}, \ \ \ \
\mathscr N ( D ( \alpha ) + k ) \mathscr N =  - ( D ( \bar \alpha ) - \bar k )^* .
\end{equation}
All these symmetries are involutions, i.e. satisfy $ A^2 = I $ with $ A : L^2_0 \to L^2_0 $. For potentials satisfying \eqref{eq:UV} these
symmetries extend to symmetries of the BM Hamiltonian \eqref{eq:defBM} 
\begin{equation}
\label{eq:quaint}
\begin{aligned} 
& \text{\bf $\mathcal P\mathcal T $ symmetry:} \ \ \ &  \mathcal P \mathcal T := 
\begin{pmatrix} 0 & \mathscr Q \\ \mathscr Q & 0 \end{pmatrix}, \\
& \text{\bf particle-hole symmetry:} \ \ \ &  \mathscr S := \begin{pmatrix} \mathscr H & 0 \\
0 & \mathscr H \end{pmatrix}, 
\\ & \text{\bf mirror symmetry:} \ \ \  &  \mathscr M := \begin{pmatrix} \ 0 & i \mathscr N \\
-i \mathscr N & 0 \end{pmatrix} .
\end{aligned} 
\end{equation}
We recall the following proposition from \cite{bz23} summarizing basic properties of the aforementioned symmetries
\begin{prop}
\label{prop:sym}
For the BM Hamiltonian given in \eqref{eq:defBM} with potentials 
satisfying \eqref{eq:UV},and $ \alpha \in \mathbb C $, $ \lambda \in \mathbb R $, we have
\begin{equation}
\label{eq:comm}
\begin{gathered} 
\mathcal P \mathcal T H ( \alpha, \lambda ) = H (  \alpha, \lambda ) \mathcal P \mathcal T , \ \ \ \  
\mathscr M H ( \alpha, \lambda) = H ( \bar \alpha, \lambda ) \mathscr M, \\
\mathscr S H ( \alpha, \lambda ) = - H ( \alpha, \lambda ) \mathscr S.
\end{gathered} 
\end{equation}
Additionally, for $ k = \{0, \pm K\} + \Lambda^*$, $ p \in \mathbb Z_3$, and 
$ K = \frac43 \pi $, 
\begin{equation}
\label{eq:mapp}
\mathcal P \mathcal T : L^2_{k,p} \to L^2_{k,-p+1}, \ \ 
\mathscr S : L^2_{k,p} \to L^2_{-k-K,p}, \ \ \mathscr M : L^2_{k,p} \to L^2_{-k-K,1-p}, 
\end{equation}
where  $L^2_{\bullet, \bullet} := L^2_{\bullet, \bullet} ( \mathbb C/\Lambda; \mathbb C^4 ) $.
\end{prop}

As a final ingredient from \cite{bz23}, we recall the existence of protected states of the BM Hamiltonian at $k=0$. (See \cite{gz} for a more abstract formulation.)
\begin{prop}
\label{prop:prote}
In the notation of \eqref{eq:floH} and \eqref{eq:defLpk} and for all $ \alpha, \lambda \in \mathbb R $, 
\begin{equation}
\label{eq:prote}
\dim \ker_{ L^2_{0,0} }  H_{0} ( \alpha , \lambda )  \geq 1 ,\ \ \dim \ker_{ L^2_{0,1} }  H_{0} ( \alpha , \lambda )  \geq 1.
\end{equation}
\end{prop}




\section{Existence of Dirac cones}
\label{sec:cone}
We prove Theorem \ref{thm:generic} in this section. We start with a general argument on perturbations of self-adjoint operators with $\ZZ/3\ZZ$ symmetry. The proof of this theorem relies on the Schur complement formula, which in the version that we use, is often referred to as a \emph{Grushin problem}. For a general discussion of Grushin problems, we refer to \cite{SZ07} and \cite[Section 2.6]{notes}.

\begin{prop}\label{prop:perturb-general}
    Suppose $H_k=H_0+kA+\bar{k}A^*:\mathcal{D}\to\mathcal{H}$ is a family of (unbounded) self-adjoint operators on a Hilbert space $\mathcal{H}$ with domain $\mathcal D$. In addition, we assume that
    \begin{itemize}
        \item $\mathcal{H}$ has an orthogonal decomposition $\mathcal{H}=\mathcal{H}_0\oplus \mathcal{H}_1\oplus\mathcal{H}_2$ such that $H_0:\mathcal{H}_j\cap \mathcal D \to \mathcal{H}_j$;
        \item $H_0$ has discrete spectrum at $0$ and there exist $\varphi\in \mathcal{H}_0$, $\psi\in\mathcal{H}_1$ (normalized) such that 
        $\ker H_0=\CC\varphi+\CC\psi$;
        \item $A$ is bounded and $A:\mathcal{H}_j\to \mathcal{H}_{j+1}$ (with the convention $\mathcal{H}_3=\mathcal{H}_0$).
    \end{itemize}
Then $H_k$ has two eigenvalues $E_{\pm 1}(k)$ near $0$ satisfying
\begin{equation}\label{eq:band}
    E_{\pm 1}(k)=\pm |\langle A\varphi, \psi\rangle||k|+\mathcal{O}(|k|^2).
\end{equation}
\end{prop}
 \begin{proof}
     Let $R_+:\mathcal{D}\to \CC^2$ and $R_-:\CC^2\to\mathcal{H}$ be defined by
     \begin{equation}
     \label{eq:Rpm}
         R_+u=(\langle u,\varphi\rangle,\langle u,\psi\rangle),\quad R_-(a,b)=a\varphi+b\psi.
     \end{equation}
     Consider the following Grushin problem
     \begin{equation}\label{eq:grushin-general-k}
        \mathcal P_k:= \begin{pmatrix}
             H_k-z&R_-\\
             R_+&0
         \end{pmatrix}:\mathcal{D}\times \CC^2\to \mathcal{H}\times \CC^2
     \end{equation}
     as a perturbation of the Grushin problem
         \begin{equation}\label{eq:grushin-general-0}
         \mathcal P_0=\begin{pmatrix}
             H_0-z&R_-\\
             R_+&0
         \end{pmatrix}:\mathcal{D}\times \CC^2\to \mathcal{H}\times \CC^2.
     \end{equation}
For sufficiently small $\varepsilon>0$, the operator $\mathcal P_0$ in \eqref{eq:grushin-general-0} is invertible for $z \notin (\Spec(H_0)\setminus \{0\})+ (-\varepsilon,\varepsilon)$  with inverse given by
 \begin{equation}
     \label{eq:P0inv}
     \mathcal P_0^{-1} =: \begin{pmatrix} E& E_+ \\ E_- & E_{-+}\end{pmatrix} 
 \end{equation}
and operators
 \begin{equation}
 \label{eq:E}
     E=(\Pi^\perp (H_0-z)\Pi^\perp)^{-1}\Pi^\perp,\quad E_+=R_-,\quad E_-=R_+,\quad E_{-+}=z.
 \end{equation}
Here, $\Pi=R_-R_+$ is the projection to $\ker H_0$ and $\Pi^{\perp}=I-\Pi$.
 The operator $\mathcal P_k$ in \eqref{eq:grushin-general-k} is invertible for $\vert k\vert$ sufficiently close to zero with inverse
 \begin{equation}
     \label{eq:Pkinv}
     \mathcal P_k^{-1} =: \begin{pmatrix} F & F_+ \\ F_- & F_{-+} \end{pmatrix}.
 \end{equation}
Here, for $\vert k\vert$ small enough
 \begin{equation}
 \label{eq:F}
 \begin{aligned}
        F_{-+}&=E_{-+}+\sum\limits_{j=0}^{\infty}(-1)^j E_-(kA+\bar{k}A^*)(E(kA+\bar{k}A^*))^{j-1}E_+\\
        &=z-R_+(kA+\bar{k}A^*)R_-+R_+(kA+\bar{k}A^*)E(kA+\bar{k}A^*)R_-+\mathcal{O}(|k|^3),
 \end{aligned}
 \end{equation}
 is invertible if and only if $H_k - z$ is invertible; see \cite[Lemma 2.10]{notes}. In other words, for all $\vert k\vert$ sufficiently small, $\det F_{-+} = 0$ if and only if $z$ is an eigenvalue of $H_k$.
 We compute
 \begin{equation}
 \label{eq:F1}
     R_+(kA+\bar{k}A^*)R_-=\begin{pmatrix}
         0& \bar{k}\langle A^*\psi ,\varphi\rangle\\
         k\langle A\varphi ,\psi\rangle &0
     \end{pmatrix}
 \end{equation}
 and
 \begin{equation}
 \label{eq:F12}
      R_+(kA+\bar{k}A^*)E(kA+\bar{k}A^*)R_-=\begin{pmatrix}
        |k|^2\langle (AEA^*+A^*EA)\varphi, \varphi\rangle& k^2\langle  AEA\psi,\varphi\rangle\\
        \bar{k}^2\langle A^*EA^*\varphi,\psi\rangle&|k|^2\langle (A^*EA+AEA^*)\psi,\psi\rangle
     \end{pmatrix}.
 \end{equation}
We have
\begin{equation*}
    F_{-+}-R_+(kA+\bar{k}A^*)R_-=\mathcal{O}(|k|^2).
\end{equation*}
Since $F_{-+}$ and $R_+(kA+\bar{k}A^*)R_-$ are both self-adjoint, their eigenvalues near zero differ by $\mathcal{O}(|k|^2)$.
Since the eigenvalues of \eqref{eq:F1} are $\pm |\langle A\varphi,\psi\rangle ||k|$, the eigenvalues of $E_{\pm 1}(k)$ are by \cite[Th.~4.10 p.~291]{kato2013perturbation} given by
 \begin{equation*}
     E_{\pm 1}(k)=\pm |\langle A\varphi,\psi\rangle ||k|+\mathcal{O}(|k|^2).\qedhere
 \end{equation*}
 \end{proof}

\begin{rem} 
We can compute more terms in the asymptotic expansion of the inverse of the Grushin problem \eqref{eq:grushin-general-k} as illustrated in \eqref{eq:F}. In our application to $H_k = H_0 + k \begin{pmatrix} 0 & 0 \\ I_{2\times 2} & 0 \end{pmatrix} + \bar k \begin{pmatrix} 0 & I_{2\times 2}  \\ 0 & 0 \end{pmatrix}$, we have an antilinear $\mathcal{PT}$-symmetry: 
\begin{equation}
    \mathcal{PT}:\mathcal{H}\to \mathcal{H},\quad \mathcal{PT}^2=I,\quad\mathcal{PT}(\mathcal{H}_0)=\mathcal{H}_1,\,\, \mathcal{PT}(\mathcal{H}_1)=\mathcal{H}_0,\,\, \mathcal{PT}(\mathcal{H}_2)=\mathcal{H}_2
\end{equation}
where the Hilbert spaces are the respective invariant subspaces of the Hamiltonian, so that 
\begin{equation*}
    H_0\mathcal{PT}=\mathcal{PT}H_0,\quad \mathcal{PT}A=A^*\mathcal{PT}.
\end{equation*}
Consequently, we can choose $\mathcal{PT}\varphi=\psi$. This immediately implies that 
\[\langle EA^*\varphi,A^*\varphi\rangle+\langle EA\varphi,A\varphi\rangle=\langle EA\psi,A\psi\rangle+\langle EA^*\psi,A^*\psi\rangle\] and the energies are given by
\begin{equation}
\begin{split}
\label{eq:expansion}
    E_{\pm 1}(k)&=\pm |k\langle A\varphi,\psi\rangle-\bar{k}^2\langle (\Pi^{\perp}H_0\Pi^{\perp})^{-1}A^*\varphi, A\psi\rangle|\\
    & \quad -|k|^2\langle (\Pi^{\perp}H_0\Pi^{\perp})^{-1}(A+A^*)\varphi, (A+A^*)\varphi\rangle+\mathcal{O}(|k|^3).
    \end{split}
\end{equation}
\end{rem}
\begin{rem}
    Under the condition $\langle A\varphi,\psi\rangle\neq 0$, the bands $E_{\pm 1}(k)$ have a conic singularity near $k=0$. We also notice when the first order term $\langle A\varphi,\psi\rangle$ vanishes, we get
\begin{equation*}
    E_{\pm 1}(k)=c_{2,\pm}(\varphi,\psi)|k|^2+\mathcal{O}(|k|^3)
\end{equation*}
where, with the notation introduced in the proof of Proposition~\ref{prop:perturb-general},
\begin{equation*}
    c_{2,\pm}(\varphi,\psi)=\pm |\langle (\Pi^{\perp}H_0\Pi^{\perp})^{-1}A^*\varphi, A\psi\rangle|-\langle (\Pi^{\perp}H_0\Pi^{\perp})^{-1}(A+A^*)\varphi, (A+A^*)\varphi\rangle.
\end{equation*}
\end{rem}

\begin{figure}[ht]
\includegraphics[width=5.4cm,height=5cm]{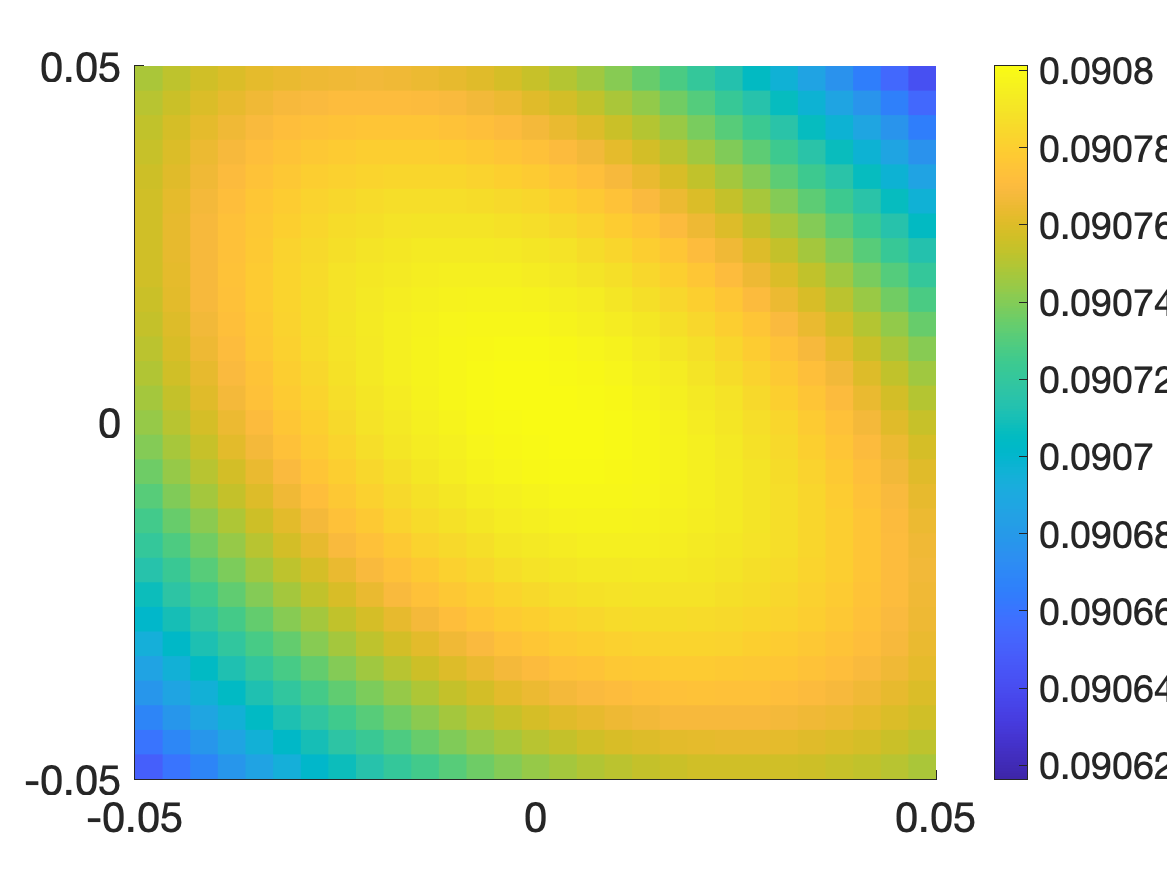}
\includegraphics[width=5.4cm,height=5cm]{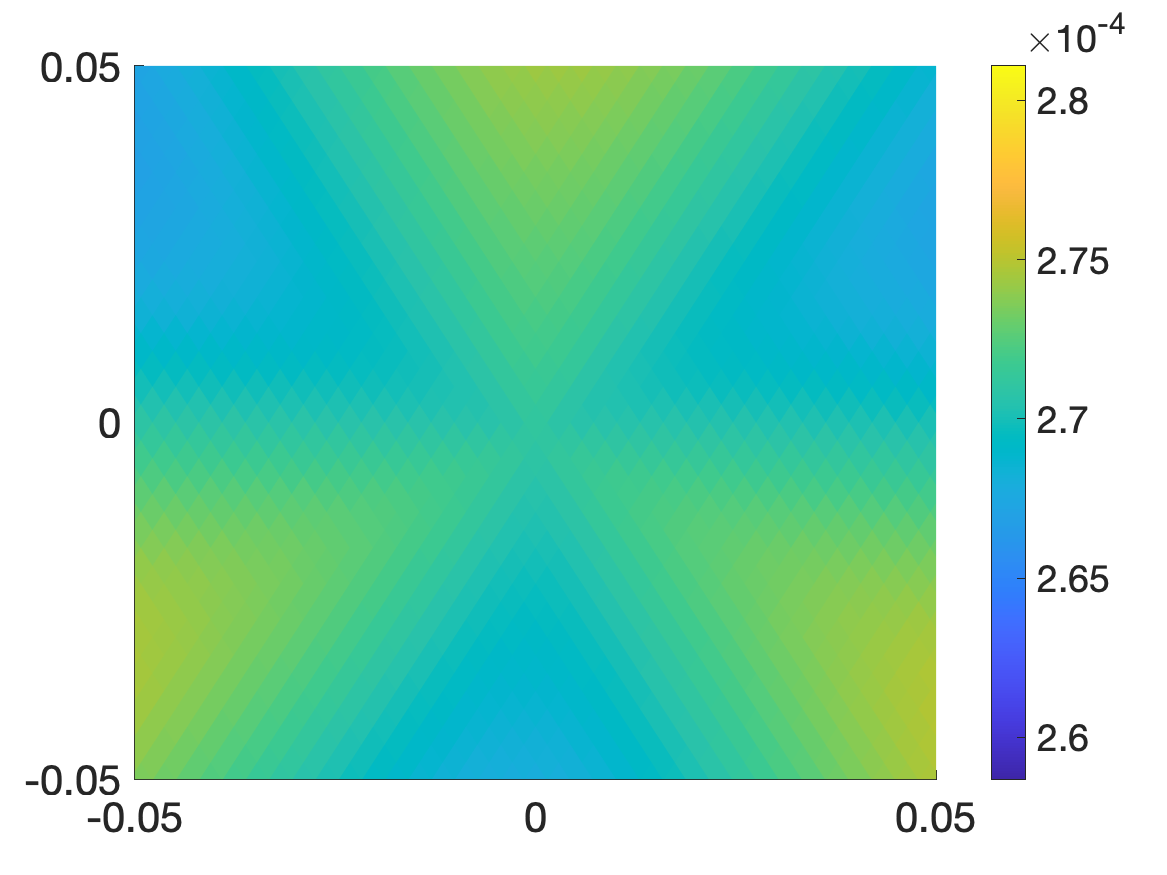}
\includegraphics[width=5.4cm,height=5cm]{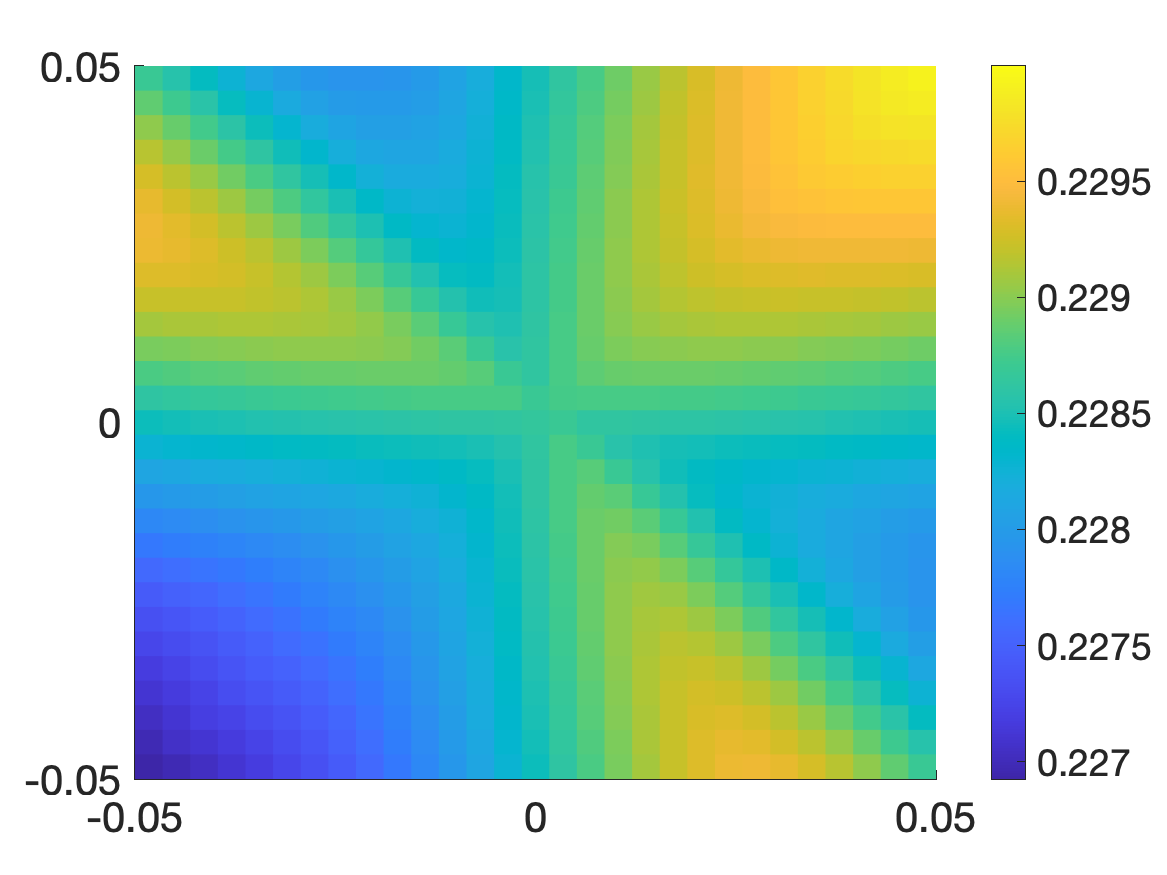}
\caption{The function $k \mapsto E_1(k)/|k|,\ k\in [-0.05,0.05]^2$ for $(\alpha,\lambda) = (0.7,0),(0.58655,0.1),(0.2,0.2)$ from left to right. The \emph{left figure} is for the chiral limit for a twisting angle at which the Dirac point is present, indicated by the fairly flat non-zero level set; in the \emph{center figure} one is close to the line of vanishing Fermi velocity, explaining the small values occurring in the figure. The \emph{right figure} is the analogue of the left figure away from the chiral limit. The color inhomogeneity in different angles suggests that the shape of the cone is affected by the higher order terms in equation \eqref{eq:band}.}
\end{figure}

We next apply Proposition \ref{prop:perturb-general} to the BM Hamiltonian \eqref{eq:defBM}. We also prove that, because of symmetries, we can assume that $\langle A\varphi,\psi\rangle$ in equation \eqref{eq:band} is real-valued. This fact will be important in the proof of Theorem \ref{thm:line}. 


\begin{prop}
\label{prop:real}
    Suppose there is an open set $\Omega\subset \RR^2$ such that 
    \begin{equation}\label{eq:condition-dim=2}        \dim \ker_{L^2_0(\CC/\Lambda;\CC^4)}\left(H_0(\alpha,\lambda)\right)=2,\quad (\alpha,\lambda)\in\Omega.
    \end{equation}
    Then there exists a real analytic function $v_F(\alpha,\lambda):\Omega\to \RR$ such that 
    \begin{equation}\label{eq:epm1-bm}
        E_{\pm 1}(\alpha,\lambda;k)=\pm |v_F(\alpha,\lambda)||k|+\mathcal{O}(|k|^2).
    \end{equation}
    In particular, when $v_F(\alpha,\lambda)\neq 0$, the bands exhibit a conic singularity at energy zero with slope $|v_F|$.
\end{prop}

\begin{proof} 
We consider the spectrum of $H_k(\alpha,\lambda)$ on $\mathcal{H}=L^2_0(\CC/\Lambda;\CC^4)$. It satisfies the assumptions of Proposition~\ref{prop:perturb-general} with 
$$\mathcal{H}_j=L^2_{0,j} \text{ and } A=\begin{pmatrix}    0&0\\
    I_{2\times 2}&0
\end{pmatrix}, \text{ and }\ker_{L^2_0(\CC/\Lambda;\CC^4)}\left(H_0(\alpha,\lambda)\right)=\CC\varphi(\alpha,\lambda)+\CC\psi(\alpha,\lambda).
$$
Here $\varphi(\alpha,\lambda)\in L^2_{0,0}(\CC/\Lambda;\CC^4)$ and $\psi(\alpha,\lambda)\in L^2_{0,1}(\CC/\Lambda;\CC^4)$ are normalized protected states. 


We first claim that the spectral projection is analytic for $(\alpha,\lambda)\in \Omega$. For any $(\alpha_0,\lambda_0)\in \Omega$, by assumption \eqref{eq:condition-dim=2}, there is a neighbourhood of $(\alpha_0,\lambda_0)$ in which we can choose $\#\Spec_{L^2_{0,j}}(H_0(\alpha,\lambda)) \cap [-\delta,\delta]=1$, $j=0,1$. The spectral projection can then be written as a contour integral
\begin{equation}
\label{eq:holomorphic}
\mathbbm{1}_{0}(H_0(\alpha,\lambda)|_{L^2_{0,j}}) = \frac{1}{2\pi i} \int_{\partial B_0(\delta)}(z-H_0(\alpha,\lambda))^{-1} dz|_{L^2_{0,j}}, \quad j=0,1.
\end{equation}
The analyticity of $H_0(\alpha,\lambda)$ implies that the spectral projection $\Pi_j = \mathbbm{1}_{0}(H_0(\alpha,\lambda)|_{L^2_{0,j}})$ is real analytic.

We now show that $\langle A\varphi,\psi\rangle$ can be chosen to be real valued. 
Suppose 
\begin{equation*}
    \varphi(\alpha,\lambda)=\begin{pmatrix}
    u(\alpha,\lambda)\\
    v(\alpha,\lambda)
\end{pmatrix},\quad \psi(\alpha,\lambda)=\begin{pmatrix}
    \tilde{u}(\alpha,\lambda)\\
    \tilde{v}(\alpha,\lambda)
\end{pmatrix},
\end{equation*}
then $\langle A\varphi,\psi\rangle= \langle u,\tilde{v}\rangle$.
Recall the particle-hole and mirror symmetries \eqref{eq:quaint} from Section \ref{sec:symmetry}.
Note that by Proposition \ref{prop:sym}, 
\begin{gather*}
    \mathscr{S}\mathscr{M}: L^2_{0,p}(\CC/\Lambda;\CC^4)\to L^2_{0,1-p}(\CC/\Lambda;\CC^4),\ \ p=0,\pm 1,\\ 
    \mathscr{S}\mathscr{M}H_k(\alpha,\lambda)=-H_{\bar{k}}(\overline{\alpha},\lambda)\mathscr{S}\mathscr{M}.
\end{gather*}
Since we are only considering $(\alpha,\lambda)\in\RR^2$, $\mathscr{SM}$ maps $\ker_{L^2_{0,0}(\CC/\Lambda;\CC^4)}\left(H_k(\alpha,\lambda)\right)$ to the other $\ker_{L^2_{0,1}(\CC/\Lambda;\CC^4)}\left(H_{\bar{k}(\alpha,\lambda)}\right)$ for $ k = \{0, \pm K\} + \Lambda^*$. 
In particular, for $k=0$, we can choose $\psi(\alpha,\lambda)=\mathscr{SM}\varphi(\alpha,\lambda)$. In other words, using the notation in \eqref{eq:quaint}
and the fact that $(\mathscr{HN})^2 = -I$,
\begin{equation*}
    \tilde{v}=-i\mathscr{HN} u,\quad \mathscr{HN}\tilde{v}=iu.
\end{equation*}
Therefore, since $\mathscr{H}$ and $\mathscr{N}$ are unitary, we obtain
\begin{equation*}
    \langle u,\tilde{v}\rangle=  \langle \mathscr{HN}u,\mathscr{HN}\tilde{v}\rangle=\langle i\tilde{v},iu\rangle=\langle \tilde{v},u\rangle
\end{equation*}
and $v_F(\alpha,\lambda)=\langle u,\tilde{v}\rangle$ is real valued. 

Finally we claim that we can choose $v_F(\alpha,\lambda)$ to be real analytic in $\Omega$. For this we write $v_F=\langle A\varphi,\psi\rangle$ using the spectral projection $\Pi_j = \mathbbm{1}_{0}(H_0(\alpha,\lambda)|_{L^2_{0,j}})$. Note $\Pi_0(\alpha, \lambda) \phi = \langle \phi, \varphi (\alpha, \lambda) \rangle \varphi (\alpha, \lambda)$. Using the $\mathscr{SM}$ symmetry $\psi = \mathscr{SM}\varphi$, we see that
$$
    v_F (\alpha, \lambda) = \langle A \varphi (\alpha, \lambda), \mathscr{SM} \varphi (\alpha, \lambda) \rangle = - \Tr _{L^2_0} \mathscr{SM} A \Pi_0 (\alpha, \lambda),
$$
as $\mathscr{SM}$ is unitary with $(\mathscr{SM})^2 =-\mathrm{Id}$. Since $\Pi_0$ is analytic in $(\alpha, \lambda) \in \Omega$ with $\mathscr{SM}$ and $A$ independent of $(\alpha, \lambda)$, we conclude that $v_F (\alpha, \lambda)$ is analytic in $\Omega$.
\end{proof}

\begin{rem}
    Note that only the choice of an analytic \emph{real valued} Fermi velocity $v_F(\alpha,\lambda)\in \RR$ relies on the symmetry $\mathscr{SM}$, whereas the form of equation \eqref{eq:epm1-bm} with only analytic complex-valued Fermi velocity $v_F(\alpha,\lambda)\in \CC$ does not rely on the $\mathscr{SM}$-symmetry. In general, $v_F(\alpha,\lambda)$ is uniquely defined up to a phase.
\end{rem}

\begin{rem}\label{rem:simple}
    The condition \eqref{eq:condition-dim=2} is satisfied when $\lambda = 0$ and $\alpha\notin\cA$, or when $\lambda=0$ and $\alpha\in\cA$ is a simple magic angle. The condition \eqref{eq:condition-dim=2} is also stable under perturbation, i.e. the set of $(\alpha,\lambda)$ satisfying \eqref{eq:condition-dim=2} is open.
\end{rem}

\begin{figure}[ht]
    \includegraphics[width=10.5cm,height=8cm]{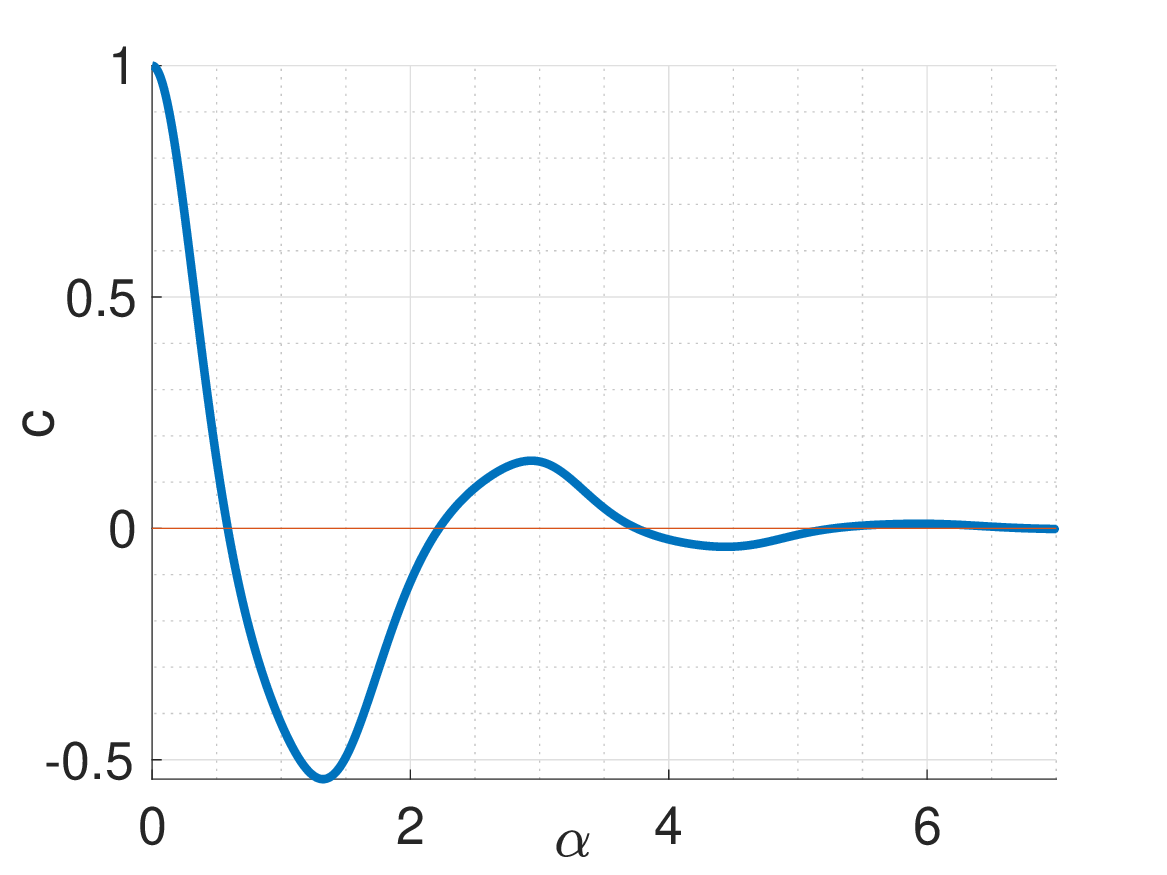}
    \caption{\label{f:real fermi}Real valued Fermi velocity $\alpha \mapsto c=v_F(\alpha,0)$ with zeros and sign changes at magic angles. The Fermi velocity satisfies the exponential estimate $v_F(\alpha,0) = \mathcal O(e^{-c_1\alpha})$ for some $c_1>0$ \cite{beta}.}
\end{figure}

The next proposition establishes that the nullspace $\ker_{L^2_0}(H_0(\alpha,\lambda))$ is generically two dimensional. This is done by identifying a real analytic function whose zero set coincides with the set where this condition fails.


\begin{prop}
\label{prop:analytic}
    There exists a locally finite family of points and analytic curves $S_i \subset \RR^2$ such that for $(\alpha,\lambda)\in \mathbb{R}^2\setminus \bigcup_{i} S_i$, $\dim\ker_{L^2_{0}} H(\alpha,\lambda) = 2$.
\end{prop}

\begin{proof}
We define the Fredholm determinant (cf.~\cite[(1.2)]{S77} or \cite[Appendix~B.5.2]{dz})
\begin{equation}
    D(\alpha,\lambda;z) = \det (I+i(H(\alpha,\lambda)+i-z)^{-3}).
\end{equation}
Note that for any $(\alpha,\lambda)\in\RR^2$, the function $D(\alpha,\lambda;z)$ has at least a zero of order two at $z=0$ by the existence of the protected state and is real analytic in $(\alpha,\lambda)$. 

Since $F(\alpha,\lambda):=\partial^2_{zz} D(\alpha,\lambda;z)\rvert_{z=0}$ is a real analytic function for $(\alpha,\lambda)\in\RR^2$ and $F(0,0)\neq 0$, by \L ojasiewicz's structure theorem (see \cite[Theorem 6.3.3]{kp}) the zeros of $F(\alpha,\lambda)$ is the union of a locally finite family of points and analytic curves $S_i\subset\RR^2$. Our proposition follows since $F(\alpha,\lambda)\neq 0$ implies $\dim\ker_{L^2_0}H_0(\alpha,\lambda)=2$.
\end{proof}

We are now in a position to \emph{almost} prove Theorem \ref{thm:generic} by the following argument. First, recall that the BM Hamiltonian exhibits a simple Dirac cone at $k = 0$ at $(\alpha,\lambda)$ if and only if $\dim \ker_{L^2_0(\CC/\Lambda;\CC^4)} \left(H_0(\alpha,\lambda)\right)=2$ and $\langle A \varphi (\alpha, \lambda), \psi (\alpha, \lambda) \rangle \ne 0$. Applying Proposition \ref{prop:analytic}, we have that the first condition fails only on a locally finite family of points and analytic curves
 $S_i \subset \RR^2$. By Proposition \ref{prop:real}, we have that the function $v_F(\alpha,\lambda)$ is real analytic everywhere outside of the sets ${S_i}$. Applying \L ojasiewicz's structure theorem (\cite[Theorem 6.3.3]{kp}) to a neighborhood of any point in $\mathbb{R}^2 \setminus \bigcup_i S_i$, we have that the zero set of $v_F(\alpha,\lambda)$ must be a finite collection of points and analytic curves.

 However, this argument does \emph{not} prove Theorem \ref{thm:generic} in full, because the zero set of $v_F$ could still accumulate at points in $\bigcup_i S_i$. In order to prove Theorem \ref{thm:generic} in full, we recall the theory of Kurdyka--Paunescu \cite{blowup}, whose results can be summarized as follows. First, eigenfunctions of operators depending analytically on parameters $(\alpha,\lambda) \in \mathbb{R}^2$ can be chosen analytically everywhere except for a locally finite set of points. In particular, eigenfunctions can generically be chosen analytically in neighborhoods of codimension 1 eigenvalue crossings. Second, at points where analytic eigenfunctions do not exist, the eigenfunctions can again be chosen analytically by lifting these points to appropiate ``blowup spaces". We can then prove Theorem \ref{thm:generic} by characterizing the zero set of $v_F$ in these blowup spaces and projecting this set down to the original parameter space.

For the reader's convenience, we recall \cite[Example 6.1]{blowup}, which demonstrates how ``blowing up" the parameter space makes it possible to choose eigenfunctions analytically at points where this is impossible otherwise. 

\begin{ex}
    Let 
    \begin{equation*}
        A(x_1,x_2)=\begin{pmatrix}
            x_1^2&x_1x_2\\
            x_1x_2&x_2^2
        \end{pmatrix}, \quad (x_1,x_2)\in\RR^2.
    \end{equation*}
    Then the eigenvalues of $A(x_1,x_2)$ are given by $\lambda_1=0$ and $\lambda_2=x_1^2+x_2^2$. The corresponding normalized eigenvectors are
\begin{align*}
    \phi_1 = \frac{1}{\sqrt{x_1^2 + x_2^2}} (x_2, -x_1), \qquad \phi_2 = \frac{1}{\sqrt{x_1^2 + x_2^2}} (x_1, x_2),
\end{align*}
which are not continuous at the origin, even though the eigenvalues are real analytic. However, if we blow up the origin in $\RR^2$, i.e., if we take $x_1 = w_1, x_2 = w_1w_2$, then the corresponding family
\begin{equation*}
        A(w_1,w_2)= w_1^2\begin{pmatrix}
            1 & w_2\\
            w_2&w_2^2
        \end{pmatrix}, \quad (w_1,w_2)\in\RR^2.
    \end{equation*}
admits a simultaneous analytic diagonalization.
\end{ex}
This motivates us to introduce the blowup space.

\begin{defi}
    The blowup space ${\rm Bl}_{(0,0)}\RR^2$ of $\RR^2$ at the point $(0,0)$ is
    \begin{equation*}
        \{((x,y),[\xi,\eta])\in\RR^2\times \mathbb{RP}^1: x\eta=y\xi\},
    \end{equation*}
    where $\mathbb{RP}^1$ is the real projective line.
    The blow down map $\pi:{\rm Bl}_{(0,0)}\RR^2\to \RR^2$ is given by
    \begin{equation*}
        ((x,y),[\xi,\eta])\mapsto (x,y).
    \end{equation*}
\end{defi}
Under this definition, for a neighbourhood $\mathcal{U}=B(0,R)$ of zero, its blowup $\td{\mathcal{U}}=\pi^{-1}(\mathcal{U})$ is given by
\[ \td{\mathcal{U}} = \{((x,y),[\xi,\eta])\in B(0,R) \times \mathbb{RP}^1: x\eta=y\xi\}, \]
which again is a two dimensional analytic manifold.
One can similarly blow up at any other point.
Following Kurdyka--Paunescu \cite{blowup}, we have the following proposition.

\begin{prop}\label{prop:blowup}
    For any point in $\RR^2$, there is a neighbourhood $\mathcal{U}$ and a blowup space $\td{\mathcal{U}}$ obtained by a finite composition of blowups such that 
    $\pi:\td{\mathcal{U}}\to \mathcal{U}$ is an isomorphism outside finitely many points $P_1,\cdots, P_r\in\mathcal{U}$, and
    there exists a normalized real analytic family of functions
    \begin{equation*}
        \varphi(u)\in L^2_{0,0},\quad \psi(u)\in L^2_{0,1},\quad u\in\td{\mathcal{U}}
    \end{equation*}
    such that for $u\notin \pi^{-1}(\{P_1,\cdots, P_r\})$,
$        \varphi(u),\psi(u)\in \ker H_0(\pi(u)).$
\end{prop}
\begin{proof}
    We may just restrict ourselves to  $L^2_{0,0}$. The analysis for $L^2_{0,1}$ is similar.
    First we reduce the problem to a finite dimensional problem. Since $H_0(\alpha,\lambda)$ is elliptic and self-adjoint, the eigenvalue at $0$ is isolated. Let $\Pi(\alpha,\lambda)$ be the spectral projector to a neighbourhood of $0$ that is analytic in a neighbourhood $\mathcal{U}$ and we consider $H_0(\alpha,\lambda)$ on the space $\Pi(\alpha,\lambda) L^2_{0,0}$. Since this vector space will in general depend on $(\alpha,\lambda)$, we may choose a basis of $\Pi(\alpha_0,\lambda_0) L^2_{0,0}$, say $e_1,\cdots, e_N$ for $(\alpha_0,\lambda_0) \in \mathcal U$. 
    As long as $\Vert \Pi(\alpha,\lambda)-\Pi(\alpha_0,\lambda_0)\Vert<1$, which locally holds by the holomorphic functional calculus, see \eqref{eq:holomorphic}, 
    we can use the Kato-Nagy formula \cite{K55}
    \begin{equation}
    \label{eq:Kato-Nagy}
    W(P,Q):=(I-(P-Q)^2)^{-1/2} (PQ + (I-P)(I-Q)),
    \end{equation}
    such that $W(P,Q) \ Q \ W(P,Q)^{-1}=P$ for $P=\Pi(\alpha,\lambda)$ and $Q = \Pi(\alpha_0,\lambda_0)$ to define $w(\alpha,\lambda):=W(\Pi(\alpha,\lambda),\Pi(\alpha_0,\lambda_0) )$
    such that locally $(w(\alpha,\lambda)e_i)_{i \in \{1,..,N\}}$ is an analytic orthonormal basis of $\operatorname{ran}(\Pi(\alpha,\lambda))$ in a perhaps smaller neighbourhood that we still denote by $\mathcal U$. This naturally defines an analytic section, see e.g.~\cite[Corr.~2.17-2.18]{notes}. We may then consider the matrix $M(\alpha,\lambda):=(\langle e_i,w(\alpha,\lambda)^*H_0(\alpha,\lambda)w(\alpha,\lambda)e_j\rangle)_{i,j}$ representing $H_0(\alpha,\lambda)$ acting on this basis. Now $M(\alpha,\lambda)$ is an analytic family of self-adjoint matrices. We can then apply \cite[Theorem 6.2]{blowup} to conclude the proposition. 
\end{proof}

\begin{proof}[Proof of Theorem \ref{thm:generic}]
    Recall that by Definition \ref{def:cone} the simple Dirac cone exists at $(\alpha,\lambda)$ if and only if the kernel $\ker_{L^2_{0}}H_0(\alpha,\lambda)$ is two dimensional and $v_{F}(\alpha,\lambda) \neq 0$. By Proposition \ref{prop:analytic}, the vector space $\ker_{L^2_{0}}H_0(\alpha,\lambda)$ is two dimensional for $(\alpha,\lambda)\notin \cup_i S_i$. Using Proposition~\ref{prop:blowup}, the function $v_F(\alpha,\lambda)$ is locally analytic for $(\alpha,\lambda)\notin \{P_1,\cdots, P_r\}$, and can be turned into an analytic function after lifting to a blowup space $\td{\mathcal{U}}$. We can then apply \L ojasiewicz's structure theorem to conclude that the zero set of $v_F$ is a locally finite union of points and analytic curves. Note that when we apply \L ojasiewicz's structure theorem in the blowup space we have to check that the blowdown of the zero set remains a finite set of points and analytic curves. But this follows directly from the definition of the blowdown map.
\end{proof}

\begin{corr}
    The set 
    \[S:=\{(\alpha,\lambda) \in \mathbb R^2;\text{ the Hamiltonian $H(\alpha,\lambda)$ does not exhibit Dirac points}\}\] has Hausdorff dimension
    \[\dim_{\mathrm{Haus}}(S)=\inf\{d\ge 0:H^d(S)=0\} \le 1,\]
    where $H^d$ is the $d$-dimensional Hausdorff measure.
\end{corr}
\begin{proof}
This follows from the previous theorem by noticing that for $d>1$
\[ H^d \bigg(\bigcup_i S_i\bigg) = \sum_i H^d (S_i) = 0.
\qedhere\]
    \end{proof}


\section{Magic angles: from chiral limit to the BM Hamiltonian}
\label{sec:magic}
In this section we prove Theorem \ref{thm:line}. By Remark~\ref{rem:simple} and Propostion~\ref{prop:real}, the Fermi velocity $v_F(\alpha,\lambda)=\langle u,\tilde{v}\rangle$ is real valued and real analytic in a neighborhood of a simple magic angle $(\alpha_0,0)$ with $\alpha_0\in\mathcal{A}$.
We refer to Table \ref{tab:c10}, where the assumption $\partial_{\alpha} v_F(\alpha_0,0) \neq 0$ is numerically verified for the first several magic angles.
We have the following expansion for the Fermi velocity.
\begin{lemm}
     Near simple magic parameters in the chiral limit $(\alpha_0,0)$ with $\alpha_0\in\mathcal{A}$, the Taylor expansion of $v_F(\alpha,\lambda)$ is given by
    \begin{equation}
    \label{eq:taylor}
    v_F(\alpha,\lambda) = c_{10}\alpha' + c_{02}\lambda^2 + c_{11}\alpha'\lambda + \mathcal{O}(\alpha'^2) + \mathcal{O}(|(\alpha',\lambda)|^3),
\end{equation}
where $\alpha':=\alpha-\alpha_0$ and the coefficients $c_{10},c_{02}$, and $c_{11}$ are given by \eqref{eq:coeff1}, \eqref{eq:coeff2}, and \eqref{eq:coeff3}, respectively.
\end{lemm}
\begin{proof}
We write the Taylor expansion
\[v_F(\alpha,\lambda) := c_0 + c_{10}\alpha' +c_{01}\lambda + c_{02}\lambda^2 + c_{11}\alpha'\lambda + \mathcal{O}(\alpha'^2) + \mathcal{O}(|(\alpha',\lambda)|^3).\]
Note that $c_0=0$ as $(\alpha_0,0)$ is a magic parameter and thus, $v_F(\alpha,0)=0$, due to the presence of a flat band. 
In a neighborhood of a simple magic parameter $(\alpha_0,0)$, by the proof of Proposition \ref{prop:real}, we can choose protected states $\varphi\in L^2_{0,0},\psi\in L^2_{0,1}$ analytic in $(\alpha,\lambda)$ such that 
\begin{gather}
 \label{eq:prot-asym}
 \varphi(\alpha_0,\lambda; z) = 
    \begin{pmatrix}
      u_0 + \sum_{i=1}^{\infty}\lambda^i u_{0,i} \\  \sum_{i=1}^{\infty} \lambda^i v_{0,i}
      \end{pmatrix}, \ \ 
    \psi(\alpha_0,\lambda; z) = 
    \begin{pmatrix}
       \sum_{i=1}^{\infty}\lambda^i \tu_{0,i} \\  \tv_0+ \sum_{i=1}^{\infty} \lambda^i \tv_{0,i}
      \end{pmatrix}.
\end{gather}
Equations 
\begin{equation}
    H_0(\alpha_0,\lambda) \varphi = H_0(\alpha_0,\lambda) \psi =0
\end{equation}
yield that, for $k\in\NN_+$, we have
\begin{gather}
    \label{eq:cond1}
    D(\alpha_0)^*v_{0,k} + Cu_{0,k-1} = 0, \ D(\alpha_0) u_0 = D(\alpha_0) u_{0,1} =0, \ D(\alpha_0)u_{0,k} + Cv_{0,k-1} = 0;\\
    \label{eq:cond2}
    D(\alpha_0)^* \tv_0 = D(\alpha_0)^* \tv_{0,1} =0,\ D(\alpha_0)^*\tv_{0,k} + C\tu_{0,k-1} = 0,\ D(\alpha_0)\tu_{0,k} + C\tv_{0,k-1} = 0.
\end{gather}
For $\alpha_0\in\mathcal{A}$, by \cite[(A.3)]{Z23} and equations \eqref{eq:cond1} and \eqref{eq:cond2}, we have 
\begin{equation}
\label{eq:coeff1}
    c_{01} = 0, \  c_{02} =  \langle \tv_0, u_{02}\rangle + \langle \tv_{02}, u_0\rangle
\end{equation}
as $u_{01}$ (resp.~$\tv_{01}$) is a scalar multiple of $u_0$ (resp.~$\tv_0$). 

To get the coefficient $c_{10}$, for $\lambda=0$, we consider $H(\alpha_0+ \alpha') \varphi = H(\alpha_0+ \alpha') \psi =0$ for $\alpha$ small with $\alpha_0\in\mathcal{A}$. By analyticity, write $\psi = (0, \tv)$ with $\tv = \sum_{i=0}^{\infty}\alpha'^i \tv_{i,0}$ and $\varphi = (u, 0)$ with $u = \sum_{i=0}^{\infty}\alpha'^i u_{i,0}$. 
Using equation $D(\alpha_0+ \alpha') u = D(\alpha_0+ \alpha')^* \tv =0$ and considering the coefficient of $\alpha'^k$ yield that
$$D(\alpha_0)u_{k,0} + B u_{k-1,0} =0, \quad B := \begin{pmatrix} 0 & U ( z )  \\
U ( -z ) & 0 \end{pmatrix}.
$$
We have 
\begin{equation}
    \label{eq:coeff2}
    c_{10} =\langle \tv_{10}, u_{0} \rangle + \langle \tv_{0}, u_{10} \rangle.
\end{equation}
Analogous computations as above show that
\begin{equation}
    \label{eq:coeff3}
    c_{11} = \langle \tv_0, u_{11} \rangle + \langle \tv_{01}, u_{10} \rangle + \langle \tv_{10}, u_{01} \rangle + \langle \tv_{11}, u_0 \rangle.
\end{equation}
Summarizing the previous computations, we have established the Taylor expansion
\begin{equation}
    v_F(\alpha,\lambda) = c_{10}\alpha' + c_{02}\lambda^2 + c_{11}\alpha'\lambda + \mathcal{O}(\alpha'^2) + \mathcal{O}(|(\alpha',\lambda)|^3),
\end{equation}
with coefficients given by equations \eqref{eq:coeff1}, \eqref{eq:coeff2}, and \eqref{eq:coeff3}.
\end{proof}

\begin{proof}[Proof of Theorem \ref{thm:line}]
    Recall that by Proposition \ref{prop:real}, $v_F(\alpha,\lambda)=\langle u,\tilde{v}\rangle\in\RR$ is real analytic in $(\alpha,\lambda)$ with $v_F(\alpha_0,0)=0$ for $\alpha_0\in\mathcal{A}\cap \RR$. By the implicit function theorem and our assumption that $\partial_{\alpha}v_F(\alpha_0,0)\neq 0$, there exists a real valued analytic function $(-\epsilon, \epsilon) \ni \lambda \mapsto \alpha(\lambda)$ with $\alpha(0)=\alpha_0$ and $\alpha'(0)=0$ such that $(\alpha,\lambda)$ is magic along the curve $ (\alpha(\lambda),\lambda)$  for $\lambda\in (-\epsilon,\epsilon)$.
    The real analyticity of the curve follows from the real analyticity of $v_F(\alpha,\lambda)$.
    The vanishing of the first derivative $\alpha'(0)=0$ follows from the Taylor expansion \eqref{eq:taylor} of $v_F(\alpha,\lambda)$, as $\partial_{\lambda}v_F(\alpha_0,0)=0$.

    Conversely, assume $\alpha_0 \notin \cA$. We then know by \cite[Appendix]{Z23} and \cite{magic} that $(\alpha_0,0)$ is non-magic in the sense of Definition \ref{def:cone}. This means $\dim \ker_{L^2_{0}(\CC/\Lambda;\CC^4)}(H_{0}(\alpha_0,0)) =2$ and $v_F(\alpha_0,0) \neq 0$. Since $v_F (\alpha, \lambda)$ and the eigenvalues of $H_{0}(\alpha,\lambda)$ are continuous in $(\alpha,\lambda)$, it follows that $v_F (\alpha, \lambda) \ne 0$ and $\dim \ker_{L^2_{0}(\CC/\Lambda;\CC^4)}(H_{0}(\alpha, \lambda)) =2$ for all $(\alpha, \lambda)$ in a sufficiently small neighborhood of $(\alpha_0, 0)$.
\end{proof}

\begin{table}[!h]
\begin{center}
\begin{tabular}{||c | c | c||} 
 \hline
 $\alpha$ & $|c_{10}|$ & $|c_{02}|$\\ [0.5ex] 
 \hline\hline
 0.58566355838955 &   1.5641 & 0.0493  \\ 
 \hline
2.2211821738201 & 0.4130 & 0.0973 \\
 \hline
 3.7514055099052 & 0.1291 & 1.4239 \\
 \hline
 5.276497782985 & 0.0355 & 9.8783 \\
 \hline
 6.79478505720 & 0.0091 & 52.5993  \\ 
 \hline
 8.3129991933 & 0.0021 & 252.5188 \\
 [1ex] 
 \hline
\end{tabular}
\caption{Values of $|c_{10}|$ and $|c_{02}|$ at first six magic angles in the Taylor expansion \eqref{eq:taylor}.}
\label{tab:c10}
\end{center}
\end{table}


\section{Magic angles in the BM Hamiltonian}
\label{sec:topo}
In this section we discuss the topology of the bands. In particular, we show that when there is a quadratic band touching, there will always be other band touching points away from the high symmetry points $-K, 0$. This is due to the topological invariance of the Euler number, which we shall introduce next. In the physics literature, the relevance of the Euler number has also been pointed out in the influential article \cite{apy19} in the context of fragile topology which inspired this section. 
\subsection{Winding number and Euler number}
We recall the notion of winding number and Euler number.
Let $\mathcal{E}\to \CC/\Lambda^*$ be an oriented rank two real vector bundle over the torus $\CC/\Lambda^*$. Suppose there is a connection on $\mathcal{E}$ with curvature
\begin{equation*}
    \Omega=\begin{pmatrix}
        0&\Omega_{12}\\
        -\Omega_{12}&0
    \end{pmatrix}.
\end{equation*}
The Euler number is
\begin{equation*}
    e_2(\mathcal{E})=\frac{1}{2\pi}\int_{\CC/\Lambda^*}\operatorname{Pf}(\Omega)=\frac{1}{2\pi}\int_{\CC/\Lambda^*} \Omega_{12},
\end{equation*}
where $\operatorname{Pf}(\cdot)$ denotes the Pfaffian of the matrix. The Euler number $e_2(\mathcal{E})$ gives a complete classification of real rank two oriented bundles $\mathcal{E}$ over the torus up to isomorphism.

We give another characterisation of the Euler number when there is a flat connection $\nabla$ on the rank two vector bundle $\mathcal{E}$ over the torus $\CC/\Lambda^*$ defined outside finitely many points. First we define the winding number of the flat connection. Suppose the connection is orthogonal, i.e., it preserves a given metric $g$ and thus satisfies $\nabla g=0$. Let $\gamma_p(t)$ be a loop around $p \in \CC /\Lambda^*$ and $e_1(k)$, $e_2(k)$ be an orthonormal local frame in a neighbourhood of $p$. Let $\theta(t)$ be the connection $1$-form under the basis $\{e_1(k), e_2(k)\}$, consider
\begin{equation*}
    \begin{pmatrix}
        0&a\\
        -a&0
    \end{pmatrix}=\frac{1}{2\pi}\oint_{\gamma_p} \theta.
\end{equation*}
The \emph{winding number} ${\rm ind}(\nabla,p)$ is defined to be the value of $a$. Its exponential
\begin{equation}\label{eq:monodromy-matrix}
\exp 2\pi\begin{pmatrix}
        0&a\\
        -a&0
    \end{pmatrix}=\begin{pmatrix}
        \cos 2\pi a&\sin 2\pi a\\
        -\sin 2\pi a&\cos 2\pi a
    \end{pmatrix}
\end{equation}
is the monodromy matrix.
From this one can then define the Euler number.
\begin{prop}\label{prop:flat-conn-euler}
    Suppose there are finitely many points $p_1,\cdots, p_\ell\in \CC/\Lambda^*$ 
    such that $\nabla$ is a flat connection on the vector bundle $\mathcal{E}$ over $(\CC/\Lambda^*)\setminus\{p_1,\cdots, p_\ell\}$. Then the Euler number is given by the sum of the winding numbers around the points $p_j$:
    \begin{equation*}
        e_2(\mathcal{E})=\sum\limits_{j=1}^{\ell}{\rm ind}(\nabla, p_j).
    \end{equation*}
\end{prop}

\begin{proof}
    We modify the flat connection smoothly in a small neighbourhood $U_j$ of the points $\{p_j\}$. Suppose the connection matrix is given by $\theta\in \Omega^1$, then the Euler number is given by
    \begin{equation*}
        e_2(\mathcal{E})=\frac{1}{2\pi}\int_{\CC/\Lambda^*}\operatorname{Pf}(d\theta) =  \frac{1}{2\pi}\sum_{j=1}^{\ell}\limits\int_{U_j}\operatorname{Pf}(d\theta),
    \end{equation*}
as curvature vanishes outside $U_j$. Let $\gamma_j = \partial U_j$ be loops around $p_j$, then Stokes' theorem gives
\begin{equation*}
    e_2(\mathcal{E})=\frac{1}{2\pi}\sum\limits_{j=1}^{\ell}\oint_{\gamma_j}\operatorname{Pf}(\theta),
\end{equation*}
where the right hand side is the sum of winding numbers ${\rm ind}(\nabla,p_j)$.
\end{proof}

\subsection{Topology of band with $\mathcal{PT}$ symmetry} As we will see, the $\mathcal{PT}$ symmetry gives rise to a non-trivial band topology in our setting that has effects on the nature of the band crossings. 

Suppose we have two bands that are separated from other bands, e.g., when we are near a simple magic angle (cf.~\cite[Theorem 2]{bhz2}). Following the standard construction we define a rank two complex vector bundle $\mathcal{E}_0$ over the torus $\CC/\Lambda^*$:
\begin{equation}
\label{eq:defE0}  
\begin{gathered}
\mathcal{E}_0 := \left\{ [ k, \phi ]_\tau  \in ( \mathbb C \times L^2_{  0 } ( \mathbb C/ \Lambda ; \mathbb C^4 ) )/
\sim_\tau : \phi \in \mathbbm{1}_{E_{\pm 1}(\alpha,\lambda,k)}( 
H_k ( \alpha,\lambda) ) \right\}, \\
( k , \phi ) \sim_\tau ( k', \phi' ) \ \Longleftrightarrow \  \exists \, p \in \Lambda^*,  \ 
k' = k + p , \ \  \phi' = \tau ( p ) \phi,  \end{gathered} \end{equation}
where $\tau(p) \phi = e^{i\langle p, z \rangle} \phi$. 
We consider the real subbundle $\mathcal{E}\subset \mathcal{E}_0$ defined by
\begin{equation*}
   \mathcal{E}=\{\varphi\in\mathcal{E}_{0}: \mathcal{PT}\varphi=\varphi\}.
\end{equation*}
This is a rank two real vector bundle such that the inner product of $\mathcal{E}_0$ restricted to $\mathcal{E}$ is real:
\begin{equation}
    \langle \varphi_1,\varphi_2\rangle=\overline{\langle \mathcal{PT}\varphi_1,\mathcal{PT}\varphi_2\rangle}=\overline{\langle \varphi_1,\varphi_2 \rangle} \text{ for all }\varphi_1,\varphi_2 \in \mathcal E.
\end{equation}

At the magic angle $\alpha_0\in\mathcal{A}$ of the chiral Hamiltonian $H_0(\alpha_0,0)$, we have flat bands and $\mathcal{E}$ can be computed explicitly. Suppose $\varphi=(u,v)^T$, then 
\begin{equation}
\label{eq:Q}
    \mathcal{PT}\varphi=\varphi\Longleftrightarrow \mathscr{Q}u=v.
\end{equation}
So $\mathcal{E}$ is defined by the equation $(D(\alpha)+k)u=0$ and equation \eqref{eq:Q}, which has Chern number $c_1(\mathcal{E})=-1$ by \cite[Theorem 4]{bhz2}. As the top Chern class $c_1(\mathcal{E})$ equals the Euler class $e_2(\mathcal{E})$ of the corresponding real bundle, in a neighbourhood of $(\alpha_0,0)$, $\mathcal{E}$ is an oriented rank two real bundle with Euler number $-1$.

Now suppose the two bands only touches at finitely many points $p_1,\cdots, p_\ell$. We may define a flat connection $\nabla$ on the bundle $\mathcal{E}$ outside the points $\{p_j\}$ by declaring that the normalized  eigenfunctions $\phi_{\pm 1}(k)$ are related by parallel transport, i.e. $\nabla \phi_{\pm 1}(k)=0$. Then by Proposition~\ref{prop:flat-conn-euler}, 
\begin{equation*}
    e_2(\mathcal{E})=\sum\limits_{j}{\rm ind}(\nabla,p_j).
\end{equation*}

It remains to compute the winding numbers near the Dirac points. For this we recall our Grushin problem \eqref{eq:grushin-general-k} with $\mathcal{PT}$ symmetry on $\mathcal{H}=L^2_{0}$. There is also a compatible $\mathcal{PT}$ symmetry on $\CC^2$ given by
\begin{equation*}
    \mathcal{PT}\begin{pmatrix}
        a\\
        b
    \end{pmatrix}=\begin{pmatrix}
        \bar{b}\\
        \bar{a}
    \end{pmatrix}.
\end{equation*}
Suppose we are in the setting of Proposition~\ref{prop:perturb-general} and we have a conic singularity at the Dirac point, i.e., $c = \langle A\varphi,\psi
\rangle\neq 0$. Then we claim the winding number is $-1/2$ (see \cite[Appendix]{bc18} for a different argument).

Using conventions in the proof of Proposition \ref{prop:perturb-general}, recall the Schur's complement formula (cf.~\cite[Lemma 2.10]{notes})
\begin{equation*}
    (H_k-z)^{-1}=F(z)-F_+(z)F_{-+}(z)^{-1}F_-(z).
\end{equation*}

We first compute the winding number at a Dirac cone. When there are Dirac cones, 
\begin{equation*}
    F_{-+}(z)=z-\begin{pmatrix}
        0&\bar{c}\bar{k}\\
        ck&0
    \end{pmatrix}+\mathcal{O}(|k|^2)
\end{equation*}
and 
\begin{equation*}
 F_{-+}(z)^{-1}=(z-E_1(k))^{-1}\Pi_{v_1(k)}+(z-E_{-1}(k))^{-1}\Pi_{v_{-1}(k)}+{\rm holomorphic\ terms},
\end{equation*}
where
\begin{equation*}
    v_{1}(k)=\frac{1}{\sqrt{2}}\begin{pmatrix}
        \left(\frac{k}{|k|}\right)^{-1/2}\\
        \left(\frac{k}{|k|}\right)^{1/2}
    \end{pmatrix}+\mathcal{O}(|k|),\quad  v_{-1}(k)=\frac{i}{\sqrt{2}}\begin{pmatrix}
        -\left(\frac{k}{|k|}\right)^{-1/2}\\
        \left(\frac{k}{|k|}\right)^{1/2}
    \end{pmatrix}+\mathcal{O}(|k|)
\end{equation*}
are normalized eigenvectors of $F_{-+}(z)$ that are invariant under $\mathcal{PT}$ symmetry, and $\Pi_{v_{\pm 1}(k)}$ are projections to the corresponding eigenvectors. Moreover, by the proof of \cite[Proposition 2.12]{notes} we have
\begin{equation*}
    F_{+}(z)=E_+(z)+\mathcal{O}(|k|),\quad F_{-}(z)=E_-(z)+\mathcal{O}(|k|),
\end{equation*}
where $E_{\pm}(z) = R_{\mp}$ is given by equations \eqref{eq:Rpm} and \eqref{eq:E}. Suppose the two bands corresponding to $v_{\pm 1}(k)$ are separated from other bands and $\Pi$ is the projection to the two bands. Then
\begin{equation*}
    \Pi=-\frac{1}{2\pi i}\oint_{\gamma} F_+(z)F_{-+}(z)^{-1}F_-(z)dz=\Pi_{v_1(k)}+\Pi_{v_{-1}(k)}+\mathcal{O}(|k|),
\end{equation*}
where $\gamma$ is a curve that encloses $z=0$ away from other bands, so that the true eigenfunctions $\phi_{\pm 1}(k)$ corresponding to the two bands are given by $\phi_{\pm 1}(k)=v_{\pm 1}(k)\cdot (\varphi,\psi)+\mathcal{O}(|k|)$ in a neighbourhood of $k=0$. Let $\varphi(k)$ and $\psi(k)$ be an orthonormal basis of $\mathbbm{1}_{E_{\pm 1}}(H_k(\alpha,\lambda))$ near $k=0$ with $\varphi(0)=\varphi$ and $\psi(0)=\psi$. Then we can write
\begin{equation*}
    \begin{pmatrix}
        \phi_{1}(k)\\
        \phi_{-1}(k)
    \end{pmatrix}= \frac{1}{\sqrt{2}}\begin{pmatrix}
         \left(\frac{k}{|k|}\right)^{-1/2}&\left(\frac{k}{|k|}\right)^{1/2}\\
          -i\left(\frac{k}{|k|}\right)^{-1/2}&i\left(\frac{k}{|k|}\right)^{1/2}
    \end{pmatrix} \begin{pmatrix}
        \varphi(k)\\
        \psi(k)
    \end{pmatrix} +\mathcal{O}(|k|).
\end{equation*}
Recall $\varphi=\mathcal{PT}(\psi)$, we can take 
\begin{equation*}
   \varphi(k)=\mathcal{PT}(\psi(k)),\quad \varphi_{+}(k)=\frac{1}{\sqrt{2}}(\varphi(k)+\psi(k)),\quad \varphi_{-}(k)=\frac{i}{\sqrt{2}}(\psi(k)-\varphi(k))
\end{equation*}
so that $\mathcal{PT}(\varphi_{\pm}(k))=\varphi_{\pm}(k)$. Then for $k=|k|e^{i\theta}$,
\begin{equation*}
        \begin{pmatrix}
        \phi_{1}(k)\\
        \phi_{-1}(k)
    \end{pmatrix}= \begin{pmatrix}
         \cos (\theta/2)&\sin(\theta/2)\\
          -\sin(\theta/2)&\cos(\theta/2)
    \end{pmatrix} \begin{pmatrix}
        \varphi_+(k)\\
        \varphi_-(k)
    \end{pmatrix} +\mathcal{O}(|k|).
\end{equation*}
In other words,
\begin{equation*}
    \begin{pmatrix}
        \varphi_+(k)\\
        \varphi_-(k)
    \end{pmatrix} = \begin{pmatrix}
         \cos (\omega(\theta))&-\sin(\omega(\theta))\\
          \sin(\omega(\theta))&\cos(\omega(\theta))
    \end{pmatrix}  \begin{pmatrix}
        \phi_{1}(k)\\
        \phi_{-1}(k)
    \end{pmatrix},\quad \omega(\theta)=\theta/2+\mathcal{O}(|k|).
\end{equation*}
Since $H_k\phi_{\pm 1}(k) = E_{\pm 1}(k)\phi_{\pm 1}(k)$, $\phi_1|_{\theta=2\pi}$ has to become $\pm \phi_1|_{\theta=0}$ after rotation around once. Therefore, $\omega(2\pi)-\omega(0)\equiv 0\mod \pi$ (this can also be seen from \eqref{eq:monodromy-matrix}). Hence $\omega(2\pi)-\omega(0)=\pi$ and the winding number is given by
\begin{equation*}
    a=\frac{1}{2\pi}\oint \langle \nabla\varphi_+(k),\varphi_-(k)\rangle = -\frac{1}{2\pi}(\omega(2\pi)-\omega(0))=-\frac{1}{2}.
\end{equation*}

Similarly, when there is a quadratic band touching, i.e., 
\begin{equation*}
    F_{-+}(k)=z-\begin{pmatrix}
        a|k|^2& b k^2\\
        \bar{b}\bar{k}^2& a|k|^2
    \end{pmatrix}+\mathcal{O}(|k|^3).
\end{equation*}
When $b=0$, $a\neq 0$, the eigenfunctions $v_{\pm 1}(k)$ are given by $(1,0)^T+\mathcal{O}(|k|)$ and $(0,1)^T+\mathcal{O}(|k|)$. So the winding number would be zero. When $b\neq 0$, the eigenfunctions are given by
\begin{equation*}
    \begin{pmatrix}
        k/|k|\\
        (k/|k|)^{-1}
    \end{pmatrix} + \mathcal{O}(|k|),\quad \begin{pmatrix}
        k/|k|\\
        -(k/|k|)^{-1}
    \end{pmatrix} + \mathcal{O}(|k|).
\end{equation*} 
The phase of $k/|k|$ would change by $2\pi$, so the winding number is $1$. We conclude the following.
\begin{prop}
    Suppose the two bands $E_{\pm 1}(k)$ are isolated from other bands and exhibit quadratic crossing at points $k=0,-K$ in the sense that, following \eqref{eq:expansion},
    \begin{equation*}
        \langle A \varphi,\psi\rangle=0,\quad \langle (\Pi^{\perp}H_0\Pi^{\perp})^{-1}A^*\varphi, A\psi\rangle\neq 0.
    \end{equation*}
    Then the two bands $E_{\pm 1}(k)$ have to touch at some addition points. Moreover if they touch at a discrete set of points, then the sum of the winding numbers of the touching points outside 
    $\{0,-K\}$ would be $-3$ or $-1$.
\end{prop}

\subsection{Wannier basis}
The nontrivial topology of the band will also give obstructions to the existence of exponential localized Wannier functions. While the Chern number of the bands around zero may be zero, the non-zero Euler number still affects the nature of the associated Wannier function, as also observed in \cite{apy19}. We have the following
\begin{prop}
Suppose the two bands $E_{\pm 1}(k)$ are separated from other bands. Then there does not exist exponential localized Wannier functions that is invariant under $\mathcal{PT}$ symmetry. More precisely, there does not exist an orthonormal family $\{\varphi_\gamma\}$ inside $\Pi_{E_{\pm 1}(k)}L^2(\CC)$ of the form
    \begin{equation}\label{eq:wannier}
        \varphi_\gamma(z)=\mathscr{L}_{\gamma}\varphi_0(z),\quad \mathcal{PT}\varphi_\gamma=\varphi_{-\gamma}, \quad \int_{\CC}|z|^2|\varphi_0(z)|^2 dz<\infty. 
    \end{equation}
\end{prop}
\begin{proof}
    This follows from \cite[Theorem 9]{notes}, see also \cite{P07} and references therein. We first observe the existence of Wannier basis satisfying \eqref{eq:wannier} implies there exists an $H^1$ orthonormal trivialization of the bundle $\mathcal{E}$. This follows from
    \begin{equation*}
        s(k,z)=\sum\limits_{\gamma\in\Lambda}e^{i\langle z+ \gamma, k\rangle}\mathscr{L}_{\gamma}\varphi_0(z).
    \end{equation*}
    One easily checks that $s$ is a unitary section satisfying $\mathcal{PT}s=s$. Since
    \begin{equation*}
        \int_{\CC/\Lambda}\int_{\CC
        /\Lambda^*}|\nabla_{k_j} s(k,z)|dk dz=C_\Lambda\int_{\CC}|x_j \varphi_0(z)|^2 dz,
    \end{equation*}
    the section we get is $H^1$. Now think of the oriented rank $2$ real vector bundle $\mathcal{E}$ as a complex line bundle, the Euler number is the same as the Chern number, which is $-1$. Since $s$ gives a unitary section of $\mathcal{E}$, we get a contradiction from \cite[Lemma 8.9]{notes}.
\end{proof}

\begin{figure}[!ht]
\includegraphics[width=8cm]{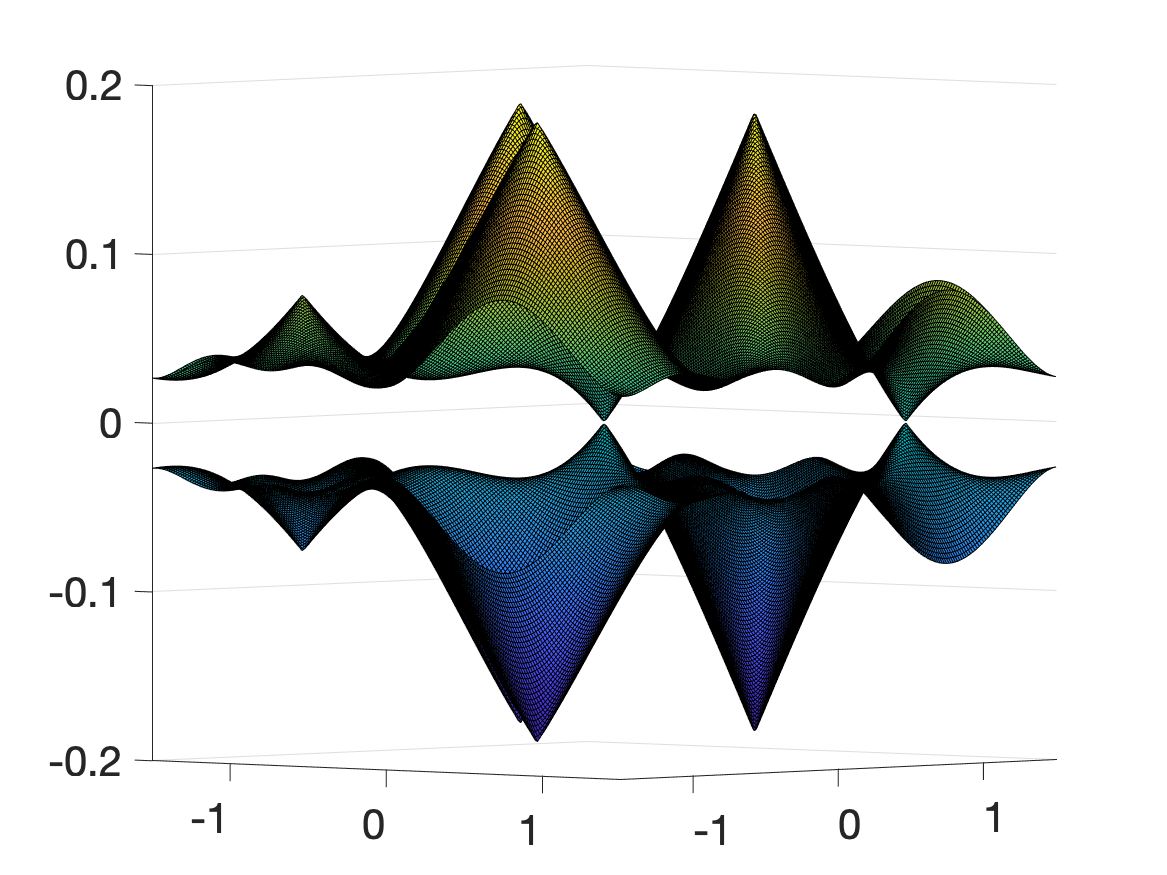} \includegraphics[width=8cm]{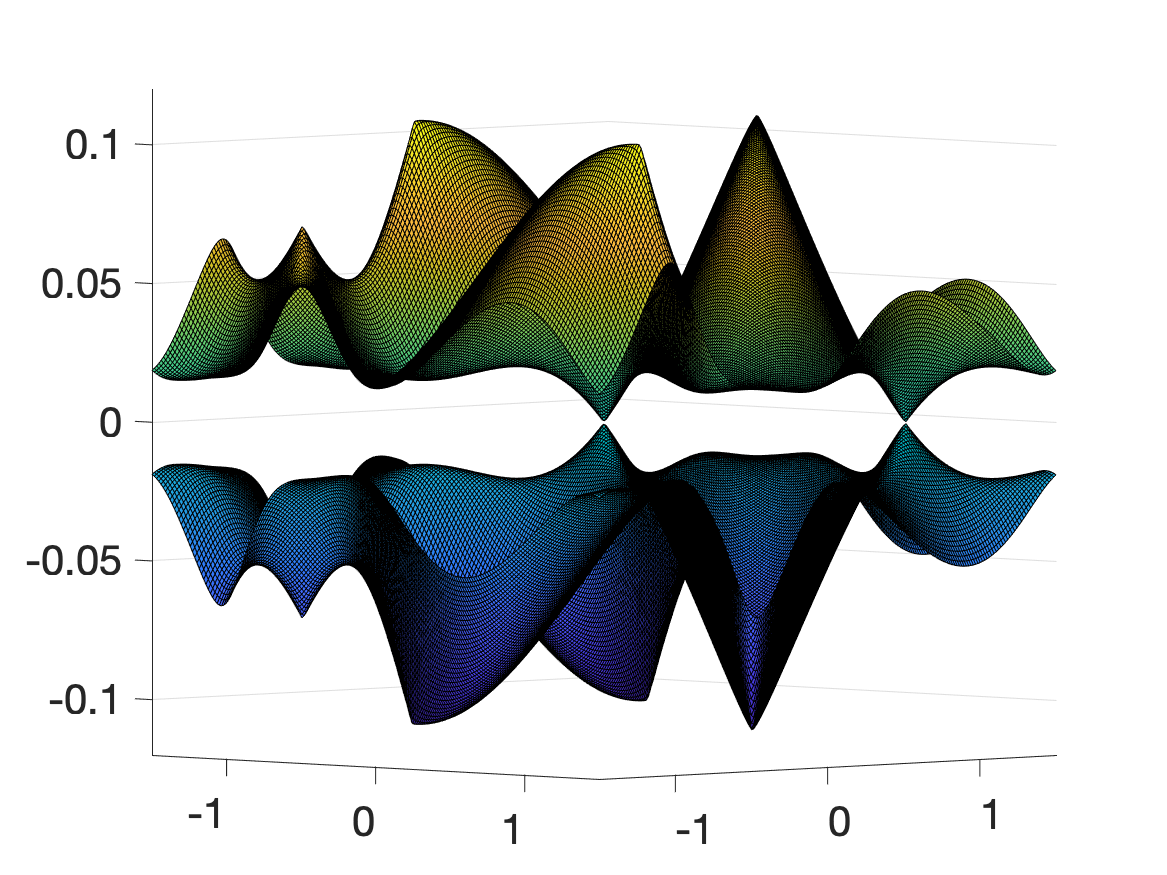}\\
\includegraphics[width=8cm]{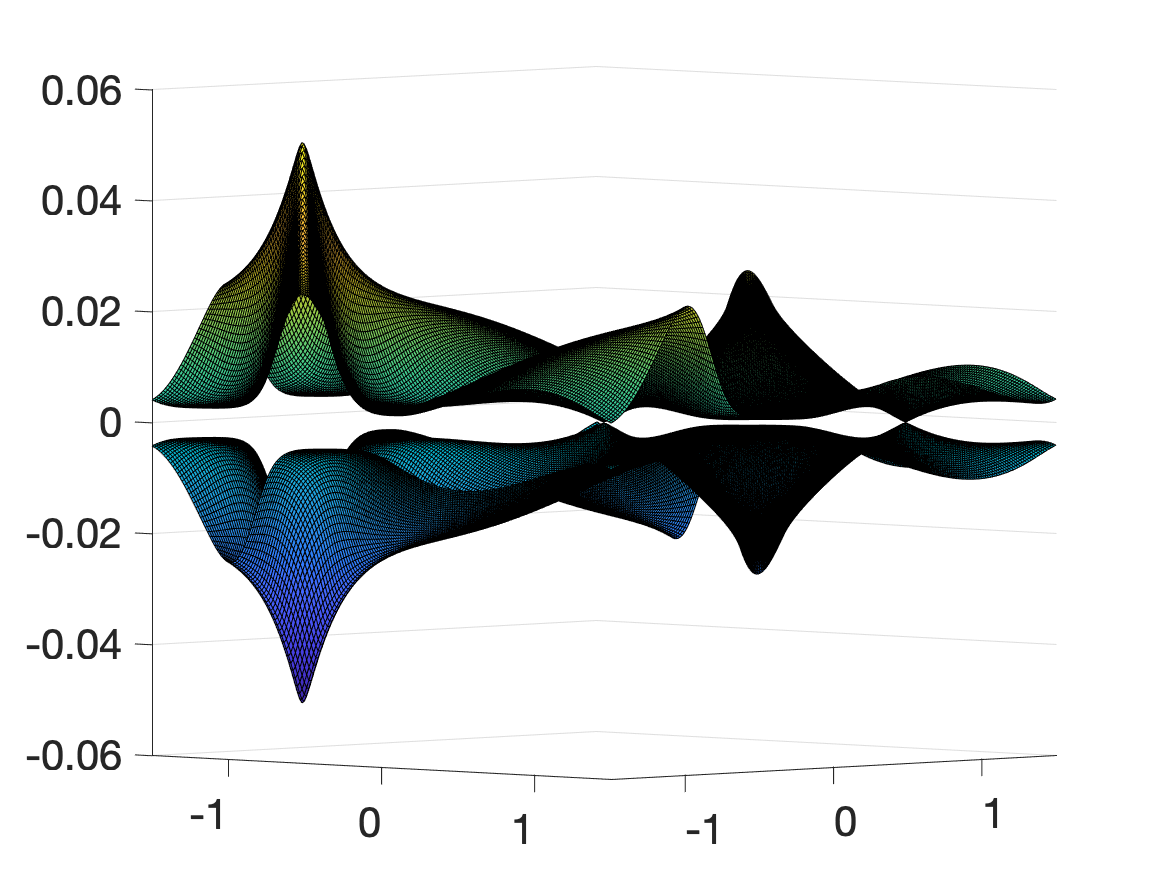}
\includegraphics[width=8cm]{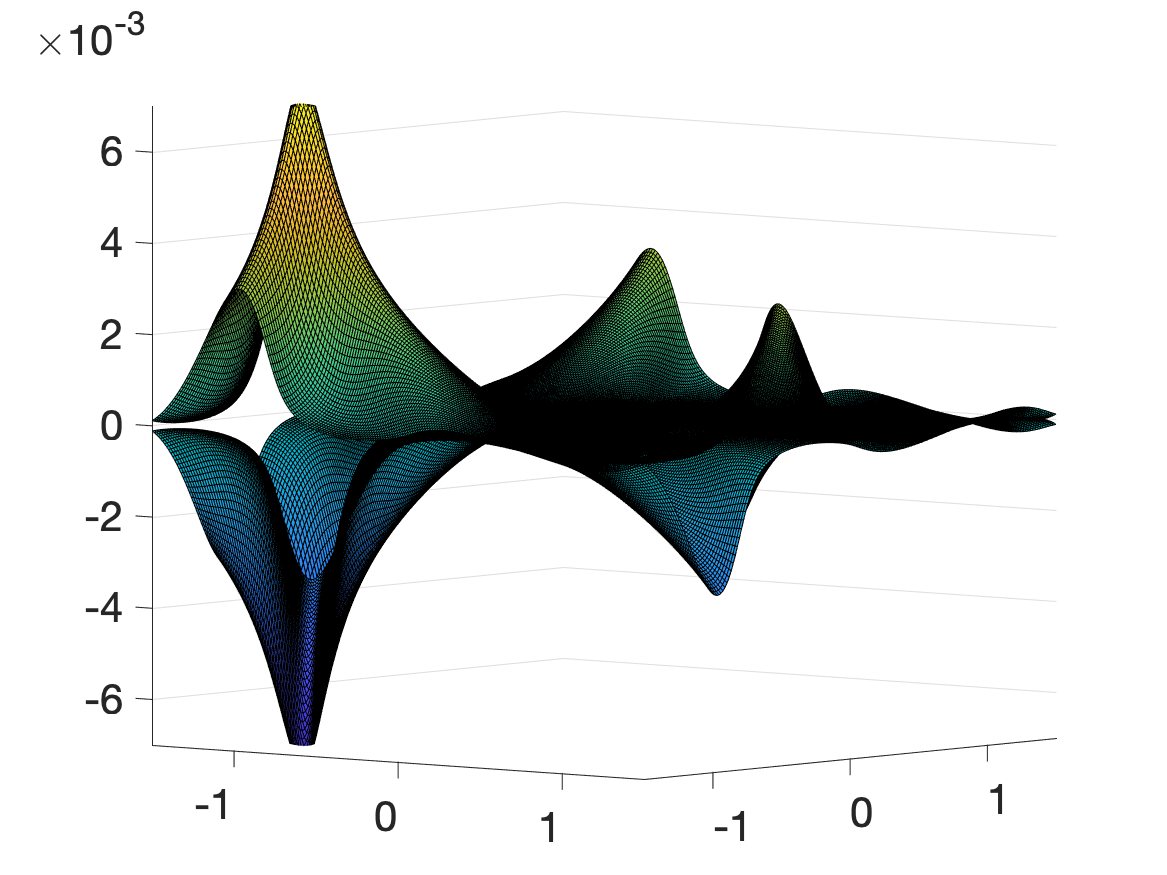}
\caption{\label{f:touching}Two bands closest to zero energy for fixed $\lambda=0.5$ and $\alpha=0$ (top left), $\alpha=0.25$ (top right), $\alpha=0.5$ (bottom left) and $\alpha=0.586$ (bottom right) as $\alpha$ approaches the line of the 1st magic angle in Figure \ref{fig:fermi}. New intersections of the bands appear.}
\end{figure}

\end{document}